\documentclass{article}

\usepackage[T1]{fontenc}
\usepackage{blindtext}
\usepackage{geometry}
\usepackage{empheq}
\usepackage{natbib}
\usepackage[figuresright]{rotating}
\usepackage{parskip}
\usepackage{enumerate}
\usepackage[shortlabels]{enumitem}
\usepackage{url}
\usepackage{amsmath,amssymb}
\usepackage{tikz}
\usepackage{caption}
\usepackage{subcaption}
\usepackage{longtable}
\usepackage{multirow}
\usepackage{datatool}
\usepackage{comment}
\usepackage{adjustbox}
\usepackage{mathtools}
\usepackage{booktabs}
\usepackage{graphicx}
\usepackage{hyperref}
\usepackage{authblk}
\usepackage{fancyhdr}
\usepackage{amsthm}
\newtheorem{theorem}{Theorem}
\newtheorem{proposition}{Proposition}
\newtheorem{lemma}{Lemma}

\theoremstyle{definition}

\newtheorem{definition}{Definition}

\DeclareMathOperator*{\opt}{opt}

\newcommand{\SOS}{\textit{\textsf{SOS1}}\xspace}

\usepackage[ruled,vlined,linesnumbered]{algorithm2e}

\usepackage[normalem]{ulem}
\usepackage{soul}
\usepackage{xcolor}

\newcommand{\rv}[1]{#1}

\definecolor{armygreen}{rgb}{0.19, 0.53, 0.43}

\definecolor{cornflowerblue}{RGB}{	100, 149, 237}

\usepackage{moresize}

\begin{document}
 

\title{\textbf{A fresh view on Least Quantile of Squares Regression based on new optimization approaches}}

\author[1]{Justo Puerto\thanks{\href{mailto:puerto@us.es}{\texttt{puerto@us.es}}}}
\author[2]{Alberto Torrejon\thanks{\href{mailto:atorrejon@us.es}{\texttt{atorrejon@us.es}}}}

\affil[1,2]{Institute of Mathematics of the University of Seville, Seville, Spain.}
\affil[1,2]{Department of Statistics and Operational Research, University of Seville, Spain.}

\maketitle

\hrulefill

\begin{abstract}


Regression analysis is an important instrument to determine the effect of the explanatory variables on response variables. When outliers and bias errors are present, the standard weighted least squares estimator may perform poorly. For this reason, many alternative robust techniques have been studied in literature. In these terms, the Least Squares Quantile (LQS), and in particular the Least Squares Median, are among the regression estimators that exhibit better robustness properties. However, the accurate computation of this estimators is computationally demanding, resulting in a difficult estimator to obtain. In this paper, new novel approaches to compute a global optimal solution for the LQS estimator based on single-level and bilevel optimization methods are proposed. An extensive computational study is provided to support the efficiency of the methods considered, and an ad hoc procedure to address the scalability of the problem to larger instances is proposed.


\end{abstract}

\hrulefill

\section{Introduction}
\label{sec:introduction}


A regression analysis models the relationship between one or more independent variables, also known as covariates or explanatory variables, and a dependent variable, also known as the response or target variable. 
Consider the following multiple linear regression model with response $y \in \mathbb{R}^n$, model matrix $\mathbf{X}_{n \times p}$, regression coefficients $\beta \in \mathbb{R}^p$ and error $\varepsilon \in \mathbb{R}^n$, namely
$y = \mathbf{X} \beta+\varepsilon.$

Classically, the error terms are assumed to be independent and identically distributed, $\varepsilon_i \sim N(0,\sigma^2), i \in \{1, \ldots, n\}$. 
Hereafter, it is assumed that $\mathbf{X}$ contains a column of ones to account for the intercept in the model.
Given data for the $i$-th sample $\left(y_i, \mathbf{x}_i\right), i \in \{1,\ldots,n\}$, and regression coefficients $\beta$, the $i$-th residual is denoted by $r_i(\beta) = y_i-\mathbf{x}_i^{\top} \beta$ for $i \in \{1, \ldots,n\}$. The traditional Least Squares $(\mathrm{LS})$ estimator is given by
\begin{equation}
\label{eq:betaLS}
\hat{\beta}^{(\mathrm{LS})} := \arg \min_{\beta} \sum_{i=1}^n r_i^2,
\end{equation}

which considers the minimization of the $\ell_2$-norm on the residuals.
Standard types of regression, such as $\mathrm{LS}$, have favorable properties if their underlying assumptions are true, but can give misleading results otherwise, i.e., are not robust to small deviations from the assumptions (such as non-constant variance, outliers or correlated data points). 
Robust regression methods provide an alternative to least squares regression by requiring less restrictive assumptions.

In this sense, $M$-estimators introduced by \cite{Huber1973} are obtained by minimizing a loss function of the residuals of the form $\sum_{i=1}^{n} \rho\left(r_{i}\right)$, where $\rho(r)$ is a symmetric function with a unique minimum at zero. 
Examples include the Huber function and the Tukey function (see \cite{Huber2011} for more details), among others.
In \cite{Koenker1978}, the Quantile Regression $(\mathrm{QR})$ estimator was introduced as a form of robust (outlier resistant) regression, see \cite{Portnoy1991}.
The $\tau$-th conditional quantile of a random variable $Y$ given $X$ is the $\tau$-th quantile of the conditional probability distribution of $Y$ given $X$, i.e, $$Q_{Y \mid X}(\tau)=\inf \left\{y: F_{Y \mid X}(y) \geq \tau\right\}$$ which is a random variable. 
In quantile regression, it is assumed that the $\tau$-th conditional quantile is given as a linear function of the explanatory variables, $$Q_{Y|X} (\tau) =X \beta.$$

Therefore, given the sample $\left(y_i, \mathbf{x}_i\right), i \in \{1, \ldots, n\}$, the estimators of the regression coefficients can be obtained as
\begin{equation}
\label{eq:betaQR}
\hat{\beta}(\tau)^{(\mathrm{QR})} := \arg \min_{\beta} \sum_{i=1}^n \rho_\tau\left(r_i \right),
\end{equation}

where, $\rho_{\tau}$ is the quantile loss function of the residuals given by
\begin{equation}
\label{eq:quantilelossfuntion_residuals}
\rho_{\tau}(r(\beta))=|r(\beta)| \left[ \tau \mathbb{I}_{\{r(\beta) \geq 0\}} + (1-\tau) \mathbb{I}_{\{r(\beta)<0\}}\right]= \begin{cases}(\tau-1) r(\beta) & \text { if } r(\beta)<0, \\ \tau r(\beta) & \text { if } r(\beta) \geq 0,\end{cases}
\end{equation}

with $0<\tau<1$, and $\mathbb{I}_{\{A\}}$ is the usual indicator function of the condition $A$, which equals 1 if $A$ is true and 0 otherwise.
Different quantile regression estimators $\hat{\beta}(\tau)$ can be obtained for different values of $\tau$. 
The minimization of the quantile loss function of the residuals leads to the convex optimization problem (\ref{eq:betaQR}), which can be transformed into a linear program by splitting the residual into its positive and negative parts, as it will be discussed in Section \ref{sec:bilevel}.
For this reason, quantile regression has many expert and intelligence systems applications, e.g., \cite{Benoit2009, Zheng2012, Reddy2013, Xu2015, Xu2017, Hesamian2019, Lyocsa2021, Wang2023, Sousa2024}.
Note that when $\tau=0.5$, the loss function $\rho_r$ is proportional to the absolute value function, and thus median regression is the same as linear regression by the Least Absolute Deviation $(\mathrm{LAD})$ estimator, introduced by \cite{Edgeworth1887}, which is given by
\begin{equation}
\label{eq:betaLAD}
\hat{\beta}^{(\mathrm{LAD})} := \arg \min _{\beta} \sum_{i=1}^n |r_i|,
\end{equation}

considering the minimization of the $\ell_1$-norm of the residuals.


Quantile regression, and hence the $\mathrm{LAD}$ estimator, like other regression $M$-estimators, can correct robustness problems due to vertical deviations (that is, related to the response variable), but not those caused by horizontal deviations (that is, related to the explanatory variables), i.e., can be highly sensitive to outliers in the covariate space, see \citet{He1990}. 
In order to control both factors, there exist many attempts in the literature that try to robustify the $\mathrm{QR}$ estimator, e.g., \citet{Adrover2004} or the Least Trimmed Quantile Regression estimator proposed by \cite{Neykov2012}. 
Beyond the above, there is a whole range of robust estimators studied in the literature whose purpose is to resist up deviant data values without becoming extremely biased.
One of the most studied robust estimators is the Least Median of Squares $\mathrm{(LMS)}$ regression estimator, studied in \citet{Hampel1975} and \citet{Rousseeuw1984}, which minimizes the median of the squared residuals, although for further purposes we similarly express it by means of the absolute residuals, i.e.,
\begin{equation}
\label{eq:betaLMS}
\hat{\beta}^{(\mathrm{LMS})} := \arg \min_{\beta} \left( \underset{i=1, \ldots, n}{\operatorname{median}} \left|r_{i}\right|\right),
\end{equation}

or its generalization by considering the $q$-th order statistic, the Least Quantile of Squares $(\mathrm{LQS})$ regression estimator, 
\begin{equation}
\label{eq:betaLQS}
\hat{\beta}^{(\mathrm{LQS})} := \arg \min _{\beta}\left|r_{(q)}\right|,
\end{equation}

where $r_{(q)}$ denotes the residual, corresponding to the $q$-th ordered absolute residual,
\begin{equation}
\label{eq:sorted_residuals}
\left|r_{(1)}\right| \leq\left|r_{(2)}\right| \leq \cdots \leq\left|r_{(n)}\right|.  
\end{equation}

Other estimators that achieve a high breakdown point and good statistical efficiency include the Least Trimmed Squares estimator (\citealp{Rousseeuw1984, Rousseeuw1987}), which minimizes the sum of squares of the $q$ smallest squared residuals, the Least Quartile Difference ($\mathrm{LQD}$) estimator, introduced by \cite{Croux1994}, and the Least Trimmed Differences (LTS) estimator, introduced by \cite{Stromberg2000}.


The majority of the regression estimators listed above belong to larger families of estimators whose descriptive, asymptotic, and behavioral properties have been well studied in the literature. 
One of particular interest to us are the $L$-estimators (\citealp{Jaeckel1972, Bickel1973}), which are linear combination of ordered statistics, see \cite{Huber2011} or \cite{Serfling2009} for further details.
In order to formally define this structure for regression purposes, let us consider $r_{ord} = ( |r_{(1)}|, |r_{(2)}|, \ldots, |r_{(n)}|)^t$ the vector with the components sorted in non-decreasing sequence, 

\begin{equation}
\label{eq:Lestimator}
\hat{\beta}^{(L)} := \arg \min_{\beta} \ \lambda^t r_{ord},
\end{equation}

where $\lambda \in \mathbb{R}^n$ is a given rank dependant vector which does not necessarily have to be positive.
Different definitions of the $\lambda$-vector also lead to different objective function based on the absolute residuals.
Table \ref{table:robust_lambdas} shows the ordered weighted formulation of different well-known robust regression estimators. 
Although the expressions of the vectors indicated in the table are straightforward, for the last two, the readers are referred to, e.g., \citet{Downton1966} or \citet{Mesa2003}.
In any case, the ordered structure presented in (\ref{eq:Lestimator}) allows great flexibility, not just to study a large number of robust estimators, but to extend the analysis to more general estimators. 
For instance,
through an appropriate $\lambda$-vector, various measures of fairness or equity and their adequacy can also be studied, see, e.g. \cite{Marsh1994} or \citet{Ogryczak2009}.
Furthermore, if quadratic terms are needed to meet the definition of the estimators, as in the case of the LS estimator 
(\ref{eq:betaLS}), one can consider extending (\ref{eq:Lestimator}) by including an additional quadratic expression of the form $r_{ord}^t M r_{ord}$, where $M \in \mathbb{R}^{n \times n}$ is a rank dependant matrix which entries are the weighting factor of the interaction between residuals. For example, by considering $M = diag_n(1,...,1)$ we meet the exact definition of (\ref{eq:betaLS}).

\begin{table}[]
\centering
\renewcommand{\arraystretch}{1}
\begin{center}
\begin{tabular}{ c c c l }
\hline
Estimator & Loss & Expression & $\lambda$-vector \\
\hline
  \small{\begin{tabular}[c]{@{}c@{}} \textbf{Least Absolute} \\ \textbf{Deviation} \\ (\citealp{Edgeworth1887}) \end{tabular}} &
  $\ell_1$-norm &
  \footnotesize $\displaystyle{\sum_{i=1}^n |r_i|}$ & 
  $(1,...,1)$ 
\\ \hline
  \small{\begin{tabular}[c]{@{}c@{}} \textbf{Minimum Maximum} \\ \textbf{Absolute Deviation} \\ (\citealp{Sposito1976}) \end{tabular}} & 
  \small{\begin{tabular}[c]{@{}c@{}} Chebyshev or \\ $\ell_\infty$-norm  \end{tabular}} &
  \footnotesize $\displaystyle{\max |r_i|}$ & 
  $(0,...,0,1)$ 
\\ \hline
  \small{\begin{tabular}[c]{@{}c@{}} \textbf{Least Median} \\ \textbf{of Squares} \\ (\citealp{Hampel1975}) \end{tabular}} &
  \small{\begin{tabular}[c]{@{}c@{}} Median \\ absolute residual \end{tabular}} &
  \footnotesize $\displaystyle{\text{median} |r_i|}$ & 
  $(0, ..., 0, \underbrace{1}_{\lfloor \frac{n}{2} \rfloor-th}, 0, ..., 0)$ 
\\ \hline
  \small{\begin{tabular}[c]{@{}c@{}} \textbf{Least Quantile} \\ \textbf{of Squares} \\ (\citealp{Rousseeuw1984})  \end{tabular}} &
  \small{\begin{tabular}[c]{@{}c@{}} $q$-th \\ absolute residual \end{tabular}} &
  \footnotesize $\displaystyle{|r_{(q)}|}$ & 
  $(0, ..., 0, \underbrace{1}_{q-th}, 0, ..., 0)$ 
\\ \hline
  \small{\begin{tabular}[c]{@{}c@{}} \textbf{Least Trimmed} \\ \textbf{of Squares} \\ (\citealp{Rousseeuw1985}) \end{tabular}} & 
  \small{\begin{tabular}[c]{@{}c@{}} $\alpha, \beta$-trimmed mean \\ ($\alpha, \beta \in [0,1]$) \end{tabular}} &
  \footnotesize $\displaystyle{\sum_{i= \lfloor \alpha n \rfloor + 1}^{n - \lfloor \beta n \rfloor} |r_{(i)} |}$ & 
  $(\underbrace{0,...,0}_{\lfloor \alpha n \rfloor}, \underbrace{1,...,1,}_{n - (\lfloor \alpha n \rfloor + \lfloor \beta n \rfloor)} \underbrace{0,...,0}_{\lfloor \beta n \rfloor})$  
\\ \hline
  \small{\begin{tabular}[c]{@{}c@{}} \textbf{Gini} \\ \textbf{Differences} \\ (\citealp{Wainer1976}, \\ \citealp{Sievers1983}) \end{tabular}} &
  \small{\begin{tabular}[c]{@{}c@{}}  Gini's mean or \\ absolute differences \end{tabular}} &
  \footnotesize $\displaystyle{  \sum_{i,j=1}^n \left|r_i-r_j\right| }$  &  
  \small{\begin{tabular}[c]{@{}c@{}} $(2(2i-n-1))_{i \in \{1,...,n\}}$  \end{tabular}}    
\\ \hline
  \small{\begin{tabular}[c]{@{}c@{}} \textbf{Least Trimmed} \\ \textbf{Differences} \\ (\citealp{Stromberg2000})  \end{tabular}} &
  \small{\begin{tabular}[c]{@{}c@{}}  $\alpha$-trimmed absolute \\ differences \\ ($\alpha \in [0,1]$) \end{tabular}} &
  \footnotesize $\displaystyle{  \sum_{i,j=1}^{\lfloor \alpha n \rfloor} \left|r_i-r_j\right| }$  &  
  $(2(2i-\lfloor \alpha n \rfloor-1)_{i \in \{1,...,\lfloor \alpha n \rfloor\}}, 0,...,0)$  
\\ \hline
\end{tabular}
\caption{Equivalent $\lambda$-vectors to express different robust estimators.}
\label{table:robust_lambdas}
\end{center}
\end{table}

Expressions in the form of (\ref{eq:Lestimator}) can be studied and generalized by means of the well-known ordered weighted average (OWA) aggregation operators, which were introduced in \cite{Yager1988}. These operators have been widely studied in the context of linear regression problems since the introduction of OWA regression by \cite{Yager2009}. To name a few works in this direction, one can find \citet{Chachi2021}, \citet{Flores2020,  Flores2022} or \cite{durso2022}.
In \cite{Blanco2018, Blanco2021}, 
the authors propose a framework to solve problems in the form of (\ref{eq:Lestimator}) based on ordered optimization models, where residuals are defined in a more general way, based on the use of the ordered median function, see \cite{Nickel2005, Puerto2019}, or the recent work in \cite{Ljubic2024}. 
Moreover, this framework guarantees that any feasible solution provides a sub-optimal estimator (in those cases where optimality can not be proven) and it can also be used to compute estimators with extra side constraints. Specifically, this leads to,  
\begin{equation}
\label{eq:Lestimator_general}
\hat{\beta} :=  \arg \min_{\beta} \ \lambda^t r_{ord}, \quad \text{subject to: } \mathbf{A} \beta \leq \mathbf{b}
\end{equation}

where $\mathbf{A}_{m \times p}, \mathbf{b}_{m \times 1}$ are given parameters in the problem representing side constraints useful to accommodate many forms on regularization on the variable $\beta$.
The counterpart to the great generality of this model, rests on the nonlinearity of the ordering operation that introduces a extra degree of complication and, therefore,
solving these models can be computationally challenging.
For this reason, tailored approaches for each particular case are needed which, for the $\mathrm{LQS}$ estimator, motivates our research in this paper.

\subsection{State-of-the-art results for Least Quantile of Squares Regression}

The $\mathrm{LQS}$, hence the $\mathrm{LMS}$, estimator is known to be computationally demanding due to the combinatorial nature of the problem, proven to be $NP$-hard by \citet{Bernholt2006}.
As a consequence, 
many different approaches in the literature to compute these estimators are not able to obtain a global minimum of the problem for large sample size instances, nor to provide certificates about the quality of the solution obtained.
In this line, one can fin the works of 
\cite{Rousseeuw1997},
\cite{Steele1986},
\cite{Stromberg1993},
\cite{Watson1998},
\cite{Agullo1997},
\cite{Giloni2002}, 
and \cite{Mount2006},
among others.

To the best of our knowledge, it is in the work of \cite{BertsimasMazumder2014} where, in addition to obtaining state-of-the-art results for medium/large instance sizes, a first-time resolution method is proposed that allows finding a provable global optimal solution for the $\mathrm{LQS}$ problem by means of robust optimization techniques; see \cite{Bertsimas2011} for further references. 
In that paper, $\mathrm{LQS}$ estimators (henceforth $\mathrm{BM}$ estimators), are obtained as solutions of a mixed integer mathematical programming model which is reviewed next for the sake of completeness. 
To this end, to model the $q$-th sorted residual, that is, $\left|r_{(q)}\right|$, the authors express the fact that $r_i \leq\left|r_{(q)}\right|$ for $q$ many absolute residuals from $\left|r_i\right|, \ldots,\left|r_n\right|$
by introducing the binary variables 
$$z_i= \begin{cases}1, & \text { if }\left|r_i\right| \leq\left|r_{(q)}\right| . \\ 0, & \text { otherwise. }\end{cases}$$

In addition, two sets of auxiliary continuous variables, $\mu_i$, $\bar{\mu}_i \geq 0$, are required such that they satisfy
$$\left|r_i\right|-\mu_i \leq\left|r_{(q)}\right| \leq\left|r_i\right|+\bar{\mu}_i, \quad i=1, \ldots, n,$$

together with the conditions that
if $\left|r_i\right| \geq\left|r_{(q)}\right|$, then $\bar{\mu}_i=0, \mu_i \geq 0$ and
if $\left|r_i\right| \leq\left|r_{(q)}\right|$, then $\mu_i=0, \bar{\mu}_i \geq 0$. 
Finally, decision variables $r_{i}^{+}, r_{i}^{-} \geq 0$, $i \in \{1, \ldots, n\}$, are introduced for the linearization of the absolute value of the residuals, $|r_i|$, such that
\begingroup
\allowdisplaybreaks
\begin{equation}
\label{eq:residuals_linearization}
\begin{aligned}
& r_{i}^{+}+r_{i}^{-} = |r_{i}|, \\
& r_{i}^{+}-r_{i}^{-} = r_i = y_{i}-\mathbf{x}_{i}^{\top} \beta, \\
& r_{i}^{+}, r_{i}^{-} \geq 0, r_{i}^{+} r_{i}^{-}=0, i \in \{1, \ldots, n \}.
\end{aligned}
\end{equation}
\endgroup

Therefore, $\mathrm{BM}$ estimators are obtained by minimizing a variable $\gamma \in \mathbb{R}$, which represents $q$-th sorted absolute residual, by means of the following mixed integer model,
\begingroup
\allowdisplaybreaks
\begin{subequations}
\label{form:BM}
\begin{align}
    \widehat{\beta}_{BM}^{(\mathrm{LQS})} := & \arg \min_{\beta} \quad \gamma,   &  \label{cons:BM_0} \\ \text{subject to: }
    & r_{i}^{+}-r_{i}^{-}=y_{i}-\mathbf{x}_{i}^{\top} \beta,   & i \in \{1, \ldots, n\},  \label{cons:BM_1} \\
    & r_{i}^{+}+r_{i}^{-}-\gamma=\bar{\mu}_{i}-\mu_{i},        & i \in \{1, \ldots, n\},  \label{cons:BM_2} \\
    & \sum_{i=1}^{n} z_{i}=q,                                  &                          \label{cons:BM_3} \\    
    & \gamma \geq \mu_{i},                                     & i \in \{1, \ldots, n\},  \label{cons:BM_4} \\
    & \mu_{i}, \bar{\mu}_{i} \geq 0,                           & i \in \{1, \ldots, n\},  \label{cons:BM_5} \\
    & r_{i}^{+}, r_{i}^{-} \geq 0,                             & i \in \{1, \ldots, n\},  \label{cons:BM_6} \\
    & (\bar{\mu}_{i}, \mu_{i}) \ : \SOS,                & i \in \{1, \ldots, n\},  \label{cons:BM_7} \\
    & (r_{i}^{+}, r_{i}^{-}): \SOS,                     & i \in \{1, \ldots, n\},  \label{cons:BM_8} \\
    & (z_{i}, \mu_{i}) \ \ : \SOS,                      & i \in \{1, \ldots, n\},  \label{cons:BM_9} \\
    & z_{i} \in\{0,1\},                                        & i \in \{1, \ldots, n\},  \label{cons:BM_10} 
\end{align}
\end{subequations}
\endgroup

where, \SOS state for the Special Ordered Sets of type 1, first introduced by \cite{Beale1970}, which are a special set of variables for which at most one of which can take a non-zero value and all others being at 0.

This approach leads to an optimal solution for any dataset where the sample data points do not necessarily have to be in general position, that is, for any subset of $\mathcal{I} \subset\{1, \ldots, n\}$ with $|\mathcal{I}|=p$, the $p \times p$ submatrix $X_{\mathcal{I}}$ has rank $p$. 
Furthermore, this framework enables us to provide a proof of the breakdown point of the LQS objective value that holds for any dataset.
The authors also propose a regularized version of the mixed-integer optimization formulation, 
which they found to be effective in accelerating the convergence of the model without any loss of accuracy in the solution, 
that may also act as a regularizer to shrink coefficients.

\textit{Our contribution.} 

Our contributions in this paper can be summarized as follows.
Due to the complexity of the LQS problem, obtaining a global optimal solution of (\ref{eq:betaLQS}) is quite demanding, mainly since it is a non-linear and non-differentiable problem.
In this paper we propose a number of novel mathematical programming formulations that improve the efficiency of the computation of the global optimal solutions for the LQS estimator.
These new formulations consider a smaller number of variables and constraints, allowing us to provide better upper bounds of the objective value in less computing time.
First, single-level formulations based on ordered optimization methods are proposed, yielding evidence of their equivalences and exploiting their properties to derive valid inequalities and bounds if needed.
Second, exploiting the properties of the quantile loss function (\ref{eq:quantilelossfuntion_residuals}), a new result is derived that allows us to propose a compact formulation with an additional constraint by resorting to the resolution of a bilevel optimization problem.
Given the equivalences between all these formulations and (\ref{form:BM}), general and robustness properties regarding the LQS estimator, such as the high breakdown point, the validity of non-general position of the data samples or the sub-optimality of the solution when optimality cannot be certified, are the same than in \cite{BertsimasMazumder2014}.
These results are supported by an extensive and detailed computational study on instances of tractable size. 
In addition, to address the scalability of the problem, i.e., the resolution of larger size instances (up to 10000 points), a novel  aggregation approach is proposed.

\textit{Structure of the paper.} 

The remainder of the paper is organized as follows. 
In Section \ref{sec:singlelevel}, novel single-level approaches for the LQS problem are described. 
Section \ref{sec:bilevel} develops the proposed bilevel optimization approaches. 
Section \ref{sec:Computational_experience} reports the computational experiments supporting our results, 
and Section \ref{sec:scale} shows the scalability of the proposed aggregation method to larger instances.
The last section contains the final remarks and conclusions.

\section{Single-level optimization approaches to Least Quantile of Squares regression}
\label{sec:singlelevel}

For understanding the remaining mathematical programming formulations, in general terms, a single-level optimization problem consists of
\begin{equation}
\notag
\label{def:single_level}
\begin{aligned}
\opt_{x \in X} & f(x) \\ \text { subject to: } 
& g(x) \geq 0, \\
& h(x) = 0,
\end{aligned}
\end{equation}

where, $\opt$ states the optimization sense of the problem, minimization or maximization of the given objective function $f: \mathbb{R}^{n} \rightarrow \mathbb{R}$, concerning the decision variables $x \in X \subseteq \mathbb{R}^n$, and possibly subject to a set of inequality constraints represented by $g: \mathbb{R}^{n} \rightarrow \mathbb{R}^{d_1}$ and/or to a set of equality constraints $h: \mathbb{R}^{n} \rightarrow \mathbb{R}^{d_2}$, where $d_1$ and $d_2$ represent the number of each type of constraints, respectively.

In this section, three alternative single-level mixed-integer optimization formulations to solve problem (\ref{eq:betaLQS}) are presented.
First, two formulations based on the ordered median function (\citealp{Nickel2005}), widely used in the Location Science context (see \citealp{Laporte2019}), are derived.
Second, extending the work by \cite{BertsimasMazumder2014}, a compact formulation based in the $k$-sum operator decomposition of the $q$-th ordered absolute residual (\citealp{Kalcsics2002, Ogryczak2003, Marin2020}) is introduced.

\subsection{Sorting constraints based formulation}
\label{subsec:sorting_based}

In the current work, given the sequence of ordered absolute residuals in (\ref{eq:sorted_residuals}), 
our modeling purposes call for the minimization of the $q$-th element of this sequence to find $\mathrm{LQS}$ estimators for problem (\ref{eq:betaLQS}).
This objective function can be formulated by means of the ordered median function via the $\lambda$-vector corresponding to the $k$-max criterion, i.e., $\lambda=(0, \ldots, 0, \underbrace{1}_{k-th}, 0, \ldots, 0)$, and stating the equivalence
\begin{equation}
\label{eq:beta_kmax_criterion}
\left|r_{(q)}\right| = \sum_{i = 1}^n \lambda_{i} | r_{(i)} | = \sum_{i = 1}^n z_{iq} |r_i|,
\end{equation}

where, the following binary decision variables are defined
$$z_{ik}= \begin{cases}1, & \text { if the $i$-th residual is in the $k$-th position of the ordered vector}, \\ 0, & \text { otherwise. }\end{cases}$$

Therefore, introducing sorting constraints for these assignment variables following the rationale in \cite{Boland2006}, a solution to the $\mathrm{LQS}$ problem is obtained as follows
\begingroup
\allowdisplaybreaks
\begin{subequations}
\label{form:SL_1}
\begin{align}
\hat{\beta}^{(\mathrm{LQS})} := & \arg \min_{\beta} \quad \sum_{i=1}^n z_{iq} |r_{i}|, & \label{cons:SL_1_0} \\ \text{subject to: }
& \sum_{i=1}^{n} z_{ik}=1,                                           & k \in \{1,\ldots,n\},      \label{cons:SL_1_1} \\
& \sum_{k=1}^{n} z_{ik}=1,                                           & i \in \{1,\ldots,n\},      \label{cons:SL_1_2} \\
& \sum_{i=1}^{n} z_{ik}|r_{i}| \leq \sum_{i=1}^{n} z_{i,k+1}|r_{i}|, & k \in \{1,\ldots,n\}, k<n, \label{cons:SL_1_3} \\
& z_{ik} \in\{0,1\},                                                 & i,k \in \{1,\ldots,n\}.      \label{cons:SL_1_4} 
\end{align}
\end{subequations}
\endgroup

Constraints (\ref{cons:SL_1_1}), (\ref{cons:SL_1_2}) and (\ref{cons:SL_1_3}) model the sorting of the absolute residuals. 
In particular, 
constraints (\ref{cons:SL_1_1}) guarantee that every order position is assigned to a residual, 
constraints (\ref{cons:SL_1_2}) ensure that every residual is assigned to a sorting position and 
constraints (\ref{cons:SL_1_3}) sort the absolute value of the residuals in non-decreasing order in the objective function.
Note that formulation (\ref{form:SL_1}) is nonlinear due to the product of a binary and a non-negative variables, in both the objective function (\ref{cons:SL_1_0}) and in the sorting constraints (\ref{cons:SL_1_3}), 
and also to the absolute value of the residuals.
The first nonlinearity can be easily linearized by introducing big-M notation via \cite{Mccormick1976} envelopes, and the latter by considering the equivalences in (\ref{eq:residuals_linearization}).

A first simplification of model (\ref{form:SL_1}) can be obtained stating that, constraints (\ref{cons:SL_1_3}) provide the sorting of the absolute residual values including them as coefficients in the right and left-hand side of the inequalities, thus weakening the formulation by introducing more nonlinear terms. 
An a priori stronger formulation than (\ref{form:SL_1}) can be obtained by considering the equivalent constraints in Lemma \ref{lem:SL_1_nocoef}.

\begin{lemma}
\label{lem:SL_1_nocoef}
The set of constraints 
\begin{equation}
\label{eq:SL_1_nocoef}
    \sum_{k > \ell} z_{jk} + \sum_{k \leq \ell} z_{ik} \leq 1, \quad i,j,\ell \in \{1, \ldots, n\},
\end{equation}

provide an alternative sorting of the absolute value of the residuals, thus being equivalent to constraints (\ref{cons:SL_1_3}).
\end{lemma}

\begin{proof}

We need to show that constraints (\ref{eq:SL_1_nocoef}) sort the absolute value of the residuals in non-decreasing order. 

To this end, if $|r_i| < |r_j|$ for two $i,j \in \{1, \ldots, n\}$, and assuming $z_{j \ell} = 1$, this is, the absolute residual $j$ is in the $\ell$-th position, we have as a consequence that $$\sum_{k \leq \ell} z_{ik}= 1 \text{ and } \sum_{k > \ell} z_{jk} = 0.$$

which is equivalent to
$$ \sum_{k > \ell} z_{jk} + \sum_{k \leq \ell} z_{ik} \leq 1,$$
for every $i,j,\ell \in \{1, \ldots, n\}$.

Conversely, the constraints in \eqref{eq:SL_1_nocoef} which are binding  determine the correct sorting of the absolute residuals, by pairwise comparison.
\end{proof}

Constraints (\ref{eq:SL_1_nocoef}) are simpler than constraints (\ref{cons:SL_1_3}) in the sense that they do not include the absolute value of the residual terms as coefficients, thus removing many nonlinear terms. 
However, in counterpart, the number of these constraints is greater than the number of constraints in (\ref{cons:SL_1_3}). 
In spite of that, these constraints can be dynamically added in a branch-and-cut procedure by means of a tailored separation algorithm, see \cite{Labbe2017}.

The simplicity of the model (\ref{form:SL_1}) is clearly shadowed because it considers $O(n^2)$ variables.
However, for our purposes it is not necessary to sort all the residuals to obtain a valid formulation. 
Since we know beforehand that we need the absolute residual in the $q$-th position, it is sufficient to consider a set of binary variables which identifies the absolute residuals such that $|r_i| \leq |r_{(q)}|$, namely $z_i \in \{0,1\}, i \in \{1,...,n\}$, and another set of binary variables, namely $\delta_i \in \{0,1\}, i \in \{1,...,n\}$, which identifies the absolute residual value in the $q$-th position, $|r_i| = |r_{(q)}|$.
Therefore, following a similar rationale that in (\ref{form:SL_1}) we can have a much simpler model by
\begingroup
\allowdisplaybreaks
\begin{subequations}
\label{form:SL_1_bis}
\begin{align}
\hat{\beta}_1^{(\mathrm{LQS})} := & \arg \min_{\beta} \quad \sum_{i=1}^n \delta_i s_i, & \label{cons:SL_1_0_bis}\\ \text{subject to: }
    & r_{i}^{+}-r_{i}^{-}=y_{i}-\mathbf{x}_{i}^{\top} \beta, & i \in \{1, \ldots, n\}, \label{cons:SL_1_2_bis} \\
    & r_{i}^{+}+r_{i}^{-}=s_i,                               & i \in \{1, \ldots, n\}, \label{cons:SL_1_3_bis} \\
    & \sum_{i=1}^{n} \delta_i=1,                             &                         \label{cons:SL_1_4_bis} \\
    & \sum_{i=1}^{n} z_i=q,                                  &                         \label{cons:SL_1_5_bis} \\
    & z_i s_i \leq \sum_{i=1}^{n} \delta_i s_i,              & i \in \{1, \ldots, n\}, \label{cons:SL_1_6_bis} \\
    & r_{i}^{+}, r_{i}^{-}, s_i \geq 0,                      & i \in \{1, \ldots, n\}, \label{cons:SL_1_7_bis} \\
    & (r_{i}^{+}, r_{i}^{-}): \SOS,                          & i \in \{1, \ldots, n\}, \label{cons:SL_1_8_bis} \\
    & \delta_i, z_i \in\{0,1\},                              & i \in \{1, \ldots, n\}. 
\end{align}
\end{subequations}
\endgroup

The objective function (\ref{cons:SL_1_0_bis}) looks for the minimization of the absolute residual in the $q$-th position.
Constraints (\ref{cons:SL_1_2_bis}) and (\ref{cons:SL_1_3_bis}) represent the residuals and the absolute value of the residuals, respectively.
Constraint (\ref{cons:SL_1_4_bis}) guarantees that only one absolute residual is selected in the $q$-th position and 
constraint (\ref{cons:SL_1_5_bis}) indicates that there must be exactly $q$ residuals which are less than or equal to the $q$-th residual.
Finally, constraints (\ref{cons:SL_1_6_bis}) ensure that there are at most $q$ residuals less than or equal to the $q$-th one.
Note that formulation (\ref{form:SL_1_bis}) is non-linear due to the product of binary and non-negative variables in both the objective function (\ref{cons:SL_1_0_bis}) and constraints (\ref{cons:SL_1_4_bis}). 
Therefore, a first efficient model to calculate the LQS estimator can be obtained by linearizing these products via McCormick linearization.

Besides, following a similar rationale that in \cite{Schnepper2019}, we can further simplify formulation (\ref{form:SL_1_bis}) by considering a continuous variable, namely $\omega \geq 0$, which will represent the $q$-th absolute residual value instead of representing it by means of binary variables. 
To this end, let us consider binary variables $z_i \in \{0,1\}$, which equal one if and only if the $i$-th absolute residual is less or equal than $\omega$, i.e., $s_i > \omega$, and 0 if the $i$-th absolute residual is $s_i \leq \omega$. Then, an LQS estimator which does not include variable products is obteined as follows,
\begingroup
\allowdisplaybreaks
\begin{subequations}
\label{form:SL_1_bis_lin}
\begin{align}
\hat{\beta}_2^{(\mathrm{LQS})} := & \arg \min_{\beta} \quad \omega, & \label{cons:SL_2_0_bis_lin} \\ \text{subject to: }
    & r_{i}^{+}-r_{i}^{-}=y_{i}-\mathbf{x}_{i}^{\top} \beta, & i \in \{1, \ldots, n\}, \label{cons:SL_2_1_bis_lin} \\
    & r_{i}^{+}+r_{i}^{-}=s_i,                               & i \in \{1, \ldots, n\}, \label{cons:SL_2_2_bis_lin} \\
    & \sum_{i=1}^{n} z_i \leq n-q,                           &                         \label{cons:SL_2_4_bis_lin} \\
    & s_i \leq \omega + M z_i,                               & i \in \{1, \ldots, n\}, \label{cons:SL_2_5_bis_lin} \\
    & r_{i}^{+}, r_{i}^{-}, s_i \geq 0,                      & i \in \{1, \ldots, n\}, \label{cons:SL_2_12_bis_lin} \\
    & \omega \geq 0,                                         & i \in \{1, \ldots, n\}, \\
    & z_i \in\{0,1\},                                        & i \in \{1, \ldots, n\},
\end{align}
\end{subequations}
\endgroup

where $M$ is a sufficiently large constant.
We know that $s_{(q)} \leq \omega$ if the number of residuals such $s_i \geq s_{(q)}$ are less o equal than $n-q$,
which is imposed by constraint (\ref{cons:SL_2_4_bis_lin}), while constraints (\ref{cons:SL_2_5_bis_lin}) guarantee that we are counting the absolute residuals which are less than or equal to $\omega$ since, if $z_i = 0$, then $s_i \leq \omega$, and if $z_i = 1$, then $s_i > \omega$ for a sufficiently large constant $M$.

We can derive valid bounds for the big-$M$ constant required in the formulation above. Let us assume that we are given upper and lower bounds, $\beta_j^{+}$ and $\beta_j^{-}$, for the $\beta$ coefficients in our model.  
Specifically, if $\beta_0$ denotes the solution obtained by a regression model, then for $\bar M=\eta \|\beta_0 \|_{\infty}$, for $\eta \in [1,2]$ (say), the solution found within a box of diameter $2\bar M$ centered at $\beta_0$ is valid for the optimal solution of the LQS problem. Then, for a given $\beta_0$ we can take $\bar M = \eta \| \beta_0 \|_{\infty}$ for $\eta \in[1,2]$. Hence, from the condition $\|\beta -\beta_0\|_{\infty} \le 2\bar M$, one can obtain:
$$ \beta_j^{-}:= - 2 \bar M - |\beta_{0_j} | \le \beta_j \le \beta_j^+ := 2 \bar M + |\beta_{0_j} |,\; \forall j=1,\ldots,n.$$

Next, recall that $s_i$ stands for the $i$-th absolute value residual. That is, $|y_i-\beta^tX_i|$. Let us denote by
$$ P_i=\{j:x_{ij}>0\} \mbox{ and } N_i=\{j:x_{ij}<0\}.$$
Therefore,
\begin{align*}
    s_i\le & |y_i-\sum_{j\in N} \beta_j x_{ij}|\\
    \le & |y_i|+|\sum_{j\in P_i} \beta_j x_{ij}| + |\sum_{j\in N_i} \beta_j x_{ij}|\\
    \le & |y_i|+|\sum_{j\in P_i} \beta_j^+ x_{ij}| + |\sum_{j\in N_i} \beta_j^- x_{ij}|.
\end{align*}

Hence, a valid upper bound for $M$ in the constraint \eqref{cons:SL_2_5_bis_lin} is:
$$ M_i= |y_i|+|\sum_{j\in P_i} \beta_j^+ x_{ij}| + |\sum_{j\in N_i} \beta_j^- x_{ij}|, \; \forall i=1,\ldots,n.$$

\subsection{$k$-Sum operator based formulation}
\label{subsec:ksums_based}

In addition to the results reviewed in Section \ref{sec:introduction}, the authors in \cite{BertsimasMazumder2014} provide a decomposition of the $q$-th ordered absolute residual based on the $k$-sum operator. 
To this end, let $x_{ord}=$ $\left(x_{(1)}, \ldots, x_{(n)}\right)$ be the vector with the components of $x \in \mathbb{R}^n$ sorted in non-decreasing order, i.e., $x_{(1)} \leq \cdots \leq x_{(n)}$. 
Then, if we denote by
\begin{equation}
\label{eq:ksums}
S_k(x)=\sum_{\ell=k}^n x_{(\ell)},  
\end{equation}

the sum of the $n-k+1$ largest values, also known as the $k$-sum operator, the $q$-th ordered absolute residual can we written as
\begin{equation}
\label{eq:LQS_to_ksums}
|r_{(q)}| = |y_{(q)}-\mathbf{x}_{(q)}^{\top} \beta| = \sum_{i=q}^n |y_{(i)}-\mathbf{x}_{(i)}^{\top} \beta| - \sum_{i=q+1}^n |y_{(i)}-\mathbf{x}_{(i)}^{\top} \beta| = S_q(r) - S_{q+1}(r).
\end{equation}

Computing (\ref{eq:ksums}) is equivalent to solve the following linear problem or its dual,
\begin{equation}
\label{form:ksum_primal.dual}
\begin{rcases}
\max & \sum_{i=1}^n d_i x_i,   \\ \text{subject to: } 
& \sum_{i=1}^n d_i=n-k+1,      \\
& 0 \leq d_i \leq 1, i \in \{1,\ldots,n\}.
\end{rcases}
=
\begin{cases}
\min & (n-k+1) t+\sum_{i=1}^n v_i,  \\ \text{subject to: } 
& t + v_i \geq x_i,  \ i \in \{1,\ldots,n\}, \\
& t\in \mathbb{R}, v_i \geq 0,     \ \, \quad i \in \{1,\ldots,n\}.
\end{cases}
\end{equation}

Note that, the primal problem in the left side is a continuous knapsack problem for which variables $d_i \in [0,1]$ can be considered discrete, i.e., $d_i \in \{0,1\}$. In addition, if $x_i\ge 0$ for all $i=1\ldots n$ then the variable $t$ can be considered nonnegative, namely $t\ge 0$.
Therefore, given (\ref{eq:LQS_to_ksums}) and these primal-dual relationships of the $k$-sum operator in (\ref{form:ksum_primal.dual}), 
$\mathrm{LQS}$ estimators are obtained as a result of the following model
\begingroup
\allowdisplaybreaks
\begin{subequations}
\label{form:single_level_3_bis}
\begin{align}
    \hat{\beta}_3^{(\mathrm{LQS})} := & \arg \min_{\beta} \quad (n-q+1)t + \sum_{i=1}^{n} v_{i} - \sum_{i=1}^{n} d_{i}s_i &  \label{cons:SL_3_0} \\ \text{subject to: }
    & r_{i}^{+}-r_{i}^{-}=y_{i}-\mathbf{x}_{i}^{\top} \beta, & i \in \{1, \ldots, n\}, \label{cons:SL_3_1_bis} \\
    & r_{i}^{+}+r_{i}^{-}=s_i,                               & i \in \{1, \ldots, n\}, \label{cons:SL_3_2_bis} \\
    & t + v_i \geq r_{i}^{+}+r_{i}^{-},                      & i \in \{1, \ldots, n\}, \label{cons:SL_3_3_bis} \\
    & \sum_{i=1}^{n} d_{i}= n-q,                             & i \in \{1, \ldots, n\}, \label{cons:SL_3_4_bis} \\
    & t, v_i \geq 0,                                         & i \in \{1, \ldots, n\}, \label{cons:SL_3_5_bis} \\
    & r_{i}^{+}, r_{i}^{-}, s_i \geq 0,                      & i \in \{1, \ldots, n\}, \label{cons:SL_3_6_bis} \\
    & (r_{i}^{+}, r_{i}^{-}): \SOS,                          & i \in \{1, \ldots, n\}, \label{cons:SL_3_7_bis} \\
    & d_{i} \in\{0,1\},                                      & i \in \{1, \ldots, n\}. \label{cons:SL_3_8_bis}
\end{align}
\end{subequations}
\endgroup

Note that this formulation is also nonlinear because of the terms $d_{i}s_i$ for $i \in \{1, \ldots, n\}$, which can be easily linearized by means of McCormick envelopes. 
Proposition \ref{th:SL_3} proves the validity of this formulation to compute the LQS estimator.
Furthermore, Proposition \ref{th:SL_valid} introduces a valid inequality to improve lower bounds for model (\ref{form:single_level_3_bis}).

\begin{proposition}
\label{th:SL_3}
Estimators $\hat{\beta}_3^{(\mathrm{LQS})}$, obtained solving (\ref{form:single_level_3_bis}), are solutions of the LQS problem (\ref{eq:betaLQS}).
\end{proposition}

\begin{proof}
Let $f(t,v)$ be the objective function of the problem $S^{dual}_q(r)$ and $g(d)$ the objective function of the problem $S^{primal}_{q+1}(r)$ where 

\begin{minipage}{0.5\textwidth}
\begin{align}
 S^{dual}_q(r): \min \; & (n-q+1)t + \sum_{i=1}^{n} v_{i} & \notag  \\ \text{subject to: } 
    & t + v_i \geq |r_i|,  & i \in \{1, \ldots, n\},  \notag \\
    & t, v_i \geq 0,       & i \in \{1, \ldots, n\},  \notag 
\end{align}   
\end{minipage}
\begin{minipage}{0.5\textwidth}
\begin{align}
 S^{primal}_{q+1}(r): \max \quad & \sum_{i=1}^{n} d_{i}|r_i|  \notag \\ \text{subject to: } 
    & \sum_{i=1}^{n} d_{i}= n-q, &        \notag \\
    & d_{i} \in [0,1],           & i \in \{1, \ldots, n\}.  \notag 
\end{align}
\end{minipage}

As a consequence of equivalence (\ref{eq:LQS_to_ksums}), we have that
\begin{equation}
\begin{aligned}
    \min_{t,v} f(t,v) - \max_{d} g(d) = \min_{t,v} f(t,v) + \min_{d} (- g(d)) = \min_{t,v,d} (f(t,v)-g(d)) = \min |r_{(q)}|
\end{aligned} 
\end{equation}
\end{proof}

\begin{proposition}
\label{th:SL_valid}
The following inequality
\begin{equation}
(n-q)t + \sum_{i=1}^{n} v_{i} \geq \sum_{i=1}^{n} d_{i} s_i
\end{equation}

is valid for the model (\ref{form:single_level_3_bis}).
\end{proposition}

\begin{proof}
In order to prove the validity of the inequality, we first explain the meaning of the optimal dual variables of problem $S^{primal}_{q}(r)$, namely $t^*, v_i^*$, $i \in \{1,...,n\}$. To this end, let $|r^*_{(i)}|$, $i \in \{1,...,n\}$, denote the i-$th$ optimal absolute residual values. 
It is well known from knapsack theory (\citealp{Martello1990}) that the following Dantzig algorithm produces  optimal primal  and dual solutions. Let $q$ be the index of the \emph{critical item} of the knapsack problem $S^{primal}_{q}(r)$, the optimal primal solution is
\begin{align}
\label{eq:radius:primal:sol}
d_i^{*} & = \begin{cases} 
1                      & \text{ if } i > q, \\ 
(n-q+1) - \sum_{i>q} 1 & \text{ if } i = q, \\ 
0                      & \text{ otherwise}.
\end{cases}
\end{align}

Therefore,
\begin{equation}
(n-q+1)t^* + \sum_{i=1}^n v_i^* = \sum_{i=1}^n |r_i^*| d_i^* = (n-q+1) |r^*(q)| +  \sum_{i>q} (|r_i^*|-|r^*_{(q)}|)
\end{equation}

and as a result, the optimal dual variables assume values,
\begin{equation}
\begin{aligned}
t^{*}               &= |r^*_{(q)}|, \\
v^{*}_{i}           &= \begin{cases}
|r_i^*|-|r^*_{(q)}| & \text { if } i > q, \\
0                   & \text{otherwise}.
\end{cases}
\end{aligned}
\end{equation}

Given this, it follows that, in the optimum,
\begin{equation}
\begin{aligned}
S_{q+1}(r) & = S_q(r) - |r_{(q)}| = (n-q+1)t + \sum_{i=1}^n v_i - |r_{(q)}| = \\
           & = (n-q)t + \sum_{i=1}^n v_i  \geq S_{q}(r) = \sum_{i=1}^{n} d_i |r_{i}| 
\end{aligned}
\end{equation}

which proves the validity of the inequality.
\end{proof}
 
\section{Bilevel optimization approaches to Least Quantile of Squares regression}
\label{sec:bilevel}

In this section, different novel bilevel optimization approaches for the $\mathrm{LQS}$ problem are introduced. 
For the sake of completeness, we first give a brief introduction about generic bilevel optimization problems. 
A bilevel optimization problem (\citealp{Bracken1973, Candler1977}) represents the relationship in a decision process involving hierarchical decision-making in which an upper-level decision maker aims to optimize its own objective function by selecting a decision that limits the lower-level decision's space, from which the lower-level decision maker needs to select its best decision in terms of its own objective function.

\begin{definition}[BOP] A bilevel optimization problem reads

\begin{equation}
\label{def:bilevel_upper_level}
\begin{aligned}
\opt_{x \in X, y \in Y} \ & F(x, y) \\
\text { subject to: } & G(x, y) \geq 0, \\
& y \in S(x),
\end{aligned}
\end{equation}

where, $S(x)$ is the set of optimal solutions of the $x$-parameterized problem
\allowdisplaybreaks
\begin{equation}
\label{def:bilevel_lower_level}
\begin{aligned}
\opt_{y \in Y} \ & f(x, y) \\
\text { subject to: } & g(x, y) \geq 0 .
\end{aligned}
\end{equation}

\end{definition}

For ease of presentation, we restrict ourselves here to present the definition just with inequality constraints. Problem (\ref{def:bilevel_upper_level}) is known as upper-level (or leader's) problem and Problem (\ref{def:bilevel_lower_level}) as lower-level (or follower's) problem, which is parameterized by the leader's decision $x$. 
Decision variables $x \in X \subseteq \mathbb{R}^{n_x}$ are the upper-level variables (or leader's decisions) and $y \in Y \subseteq \mathbb{R}^{n_y}$ are lower-level variables (or follower's decisions). 
The sets $X \subseteq \mathbb{R}^{n_x}$ and $Y \subseteq \mathbb{R}^{n_y}$ are typically used to denote the integrality constraints. 
For instance, $Y=\mathbb{Z}^{n_y}$ makes the lower-level problem an integer program. 
The objective functions are given by $F, f: \mathbb{R}^{n_x} \times \mathbb{R}^{n_y} \rightarrow \mathbb{R}$, and the constraint functions by $G: \mathbb{R}^{n_x} \times \mathbb{R}^{n_y} \rightarrow \mathbb{R}^m$ and $g: \mathbb{R}^{n_x} \times \mathbb{R}^{n_y} \rightarrow \mathbb{R}^{\ell}$.
Instead of using the point-to-set mapping $S$ one can also use the so-called optimal-value function
$$\varphi(x):=\min _{y \in Y}\{f(x, y): g(x, y) \geq 0\}$$

and re-write (\ref{def:bilevel_upper_level}) as
\begin{equation}
\begin{aligned}
\opt_{x \in X, y \in Y} \ & F(x, y) \\
\text { subject to: } & G(x, y) \geq 0, \\
& g(x, y) \geq 0, \\
& f(x, y) \leq \varphi(x),
\end{aligned}
\end{equation}

which is a single-level problem known as the optimal-value-function or value-function reformulation that can be solved by state-of-the-art general-purpose solvers. However, the problem with this approach is that the evaluation of the optimal value function is rather expensive. 

Indeed, the most frequent solution method for solving BOPs in practice, and the one that will be used in the sections to follow, is to reformulate the model as a single-level problem. 
Another single-level reformulation can be obtained by means of Karush-Kuhn-Tucker (KKT) conditions (\citealp{KKT1939, KKT1951}) of the lower-level problem.
This reformulation is given by the constraints of the upper-level problem and primal feasibility, dual feasibility and complementarity conditions of the lower level.
One can also use equivalently strong-duality theorem for the lower-level problem, where the single-level reformulation is given as in the KKT reformulation but considering the equivalence of the primal and dual objective functions instead of the complementarity conditions.
In practice, the strong-duality reformulation usually reports better results. 
The interested reader is referred to the survey on exact methods for bilevel otimization by \cite{Kleinert2021} for further details.

\subsection{Quantile loss and $k$-sum problems}

Recall that the quantile loss function, already introduced in Section \ref{sec:introduction}, can be written as the following piecewise linear function
$$\rho_\tau(u)=u(\tau-\mathbb{I}_{\{u<0\}})= (\tau-1) u \mathbb{I}_{\{u \leq 0\}}+\tau u \mathbb{I}_{\{u>0\}} = \begin{cases} u \tau & \text { if } u \geq 0, \\ u(\tau-1) & \text { if } u<0.
\end{cases}$$

If $Y \sim F_Y$ is a random variable, let $q_{\tau}$ be the solution to the optimization problem 
$$q_{\tau} = \arg \min_{q \in \mathbb{R}} E \big[ \rho_\tau(Y-q) \big] = \arg \min_{q \in \mathbb{R}} \left\{(\tau-1) \int_{-\infty}^q (y-q) d F_Y(y)+\tau \int_q^{\infty}(y-q) d F_Y(y)\right\}.$$

Then, the stationary condition obtained equating the derivative to zero, is
$$0=(\tau-1) \int_{-\infty}^q d F_Y(u)+\tau \int_q^{\infty} d F_Y(u)=(\tau-1) F_Y(q)+\tau(1-F_Y(q)).$$

This equation reduces to
$$F_Y\left(q \right)=\tau,$$

from which we conclude that $q_{\tau}=F_Y^{-1}(\tau)$ is the $\tau$-quantile of $F_Y$. 
Then, for a sample of size $n$, the $\tau$ sample quantile can be obtained by solving the minimization problem
\begin{equation}
\label{eq:quantile_formula}    
\widehat{q}_\tau=\arg \min_{\gamma \in \mathbb{R}} \sum_{i=1}^n \rho_\tau (y_i-\gamma) = \arg \min_{\gamma \in \mathbb{R}} L_{\tau}(y,\gamma),
\end{equation}

where
\begin{equation}
\label{eq:tau_sample_quantile_function}
L_{\tau}(y,\gamma) = (\tau-1) \sum_{y_i<\gamma} (y_i-\gamma)+\tau \sum_{y_i \geq \gamma}(y_i-\gamma).
\end{equation}

This $L_{\tau}$ is known as the $\tau$ sample quantile function for $y = (y_1,..,y_n)$.
It is well known in the literature that quantiles given by (\ref{eq:quantile_formula}) can be computed solving an optimization problem, considering two non-negative variables $u_i, v_i \geq 0$, $i\in \{1,...,n\}$, namely,
\begingroup
\allowdisplaybreaks
\begin{subequations}
\label{form:quantile_opt_problem}
\begin{align}
& \min_{\gamma \in \mathbb{R}} \quad \tau \sum_{i=1}^n u_i + (1-\tau) \sum_{i=1}^n v_i, & \label{cons:quant_0} \\ \text{subject to: }
& \gamma - y_i \leq u_i,                          & i \in \{1, \ldots, n\},          \label{cons:quant_1} \\
& y_i - \gamma \leq v_i,                          & i \in \{1, \ldots, n\},          \label{cons:quant_2} \\ 
& u_i, v_i \geq 0,                                & i \in \{1, \ldots, n\}.          \label{cons:quant_3}
\end{align}
\end{subequations}
\endgroup

Or equivalently,
\begingroup
\allowdisplaybreaks
\begin{subequations}
\label{form:quantile_opt_problem_one}
\begin{align}
& \min_{\gamma \in \mathbb{R}} \quad \tau \sum_{i=1}^n u_i + (1-\tau) \sum_{i=1}^n v_i, & \label{cons:quant_0_one} \\ \text{subject to: }
& u_i -v_i + \gamma = y_i,                     & i \in \{1, \ldots, n\},          \label{cons:quant_1_one} \\
& u_i, v_i \geq 0,                                & i \in \{1, \ldots, n\}.          \label{cons:quant_3_one}
\end{align}
\end{subequations}
\endgroup

There exists an equivalence between the $k$-sum operator (\ref{eq:ksums}) and the sample quantile function (\ref{eq:tau_sample_quantile_function}) as observed in Theorem \ref{prop:ksum_quantloss}.
\begin{theorem}
\label{prop:ksum_quantloss}
Let $y = (y_1, \ldots, y_n) \in \mathbb{R}^n$ and $\tau \in [0,1]$, then we have that 
\begin{equation}
\label{eq:ksum_quantloss}
L_{\tau}(y, \gamma) + (1-\tau) \sum_{i = 1}^n y_i \geq S_{\rv{q}} (y)
\end{equation}

for $\gamma = y_{(q)}$, the $q$-th value in the ordered sequence of $y$, with $q \geq \lceil \tau n \rceil$. 
The equality holds when $\tau n \in \mathbb{Z}_+$.
\end{theorem}

\begin{proof}
First, it is straightforward that 
\begin{equation*}
L_{\tau}(y,\gamma) = (1-\tau) \sum_{y_i \leq \gamma} (\gamma - y_i) + \tau \sum_{y_i > \gamma}(y_i-\gamma),
\end{equation*}

Therefore, it follows that
\begin{align*}
L_{\tau}(y,\gamma) + (1-\tau) \sum_{i=1}^n y_i 
& = (1-\tau)\sum_{y_i \leq \gamma} \gamma - (1-\tau) \sum_{y_i \leq \gamma} y_i + \tau \sum_{y_i > \gamma} y_i - \tau \sum_{y_i > \gamma} \gamma + (1-\tau) \sum_{i = 1}^n y_i  \\
& = (1-\tau) \sum_{y_i \leq \gamma} \gamma + \sum_{y_i > \gamma} y_i - \tau \sum_{y_i > \gamma} \gamma  \\
& = (1-\tau) \sum_{y_i \leq \gamma} \gamma + \sum_{y_i > \gamma} (y_i - \gamma) + (1- \tau) \sum_{y_i > \gamma} \gamma  \\
& = (1-\tau) \sum_{i = 1}^n \gamma + \sum_{y_i > \gamma} (y_i - \gamma)  \\ 
& = (1-\tau) \sum_{i = 1}^n \gamma + \sum_{\rv{y_i \geq \gamma}} (y_i - \gamma)  \\ 
& = (1-\tau) \gamma n + \sum_{\rv{y_i \geq \gamma}} y_i - \sum_{\rv{y_i \geq \gamma}} \gamma.
\end{align*}

Let $\gamma = y_{(q)}$ be the $q$-th value in the ordered sequence of $y$.
Then, we have that 
\begin{align*}
&  \sum_{y_i \geq \gamma} 1 = n - q \rv{+ 1,} \text{ and } \\
&  \sum_{y_i \geq \gamma} y_i = \sum_{\ell = q}^n y_{(\ell)} = S_{\rv{q}}(y).
\end{align*}

Consequently, 
\begin{align*}
L_{\tau}(y,\gamma) + (1-\tau) \sum_{i=1}^n y_i 
& = (1-\tau) \gamma n +  S_{q}(y) - \gamma (q) = \gamma (q - \tau n) +  S_{q}(y) \geq  S_{\rv{q}}(y) 
\end{align*}

for $q \geq \lceil \tau n \rceil$, and the equality holds when $\tau n \in \mathbb{Z}_+$, as we wanted to show.

\end{proof}

As a result of Theorem \ref{prop:ksum_quantloss}, another single-level formulation following a rationale similar to that of Section \ref{subsec:ksums_based} can be introduced. 
However, this reformulation will encounter the same problems as the one above with respect to the products of variables in the objective function, and therefore, it is not of interest to analyze it further.

\subsection{Formulations}

The bilevel approaches to the $\mathrm{LQS}$ problem described in this section relies on the following observation.
\begin{lemma}
\label{lem:bilevel_lemma}
Given $\tau \in [0,1]$, the equivalence
\begin{equation}
\label{eq:bilevel_equivalence}
  \min _{\beta} s_{(q)} = \min _{\beta} \left( \min_{\gamma} L_{\tau}(s, \gamma) \right),
\end{equation}

holds for $q = \lceil \tau n \rceil$, where $s_i = |r_i|$, for all $i\in \{1,\ldots,n\}$.
\end{lemma}

\begin{proof}
In order to provide the equivalence above, we only have to show that given the values of the absolute residuals $\bar s = ( s_i )_{i\in \{1,...,n\}}$ fixed, then the solution to the inner minimization problem in (\ref{eq:bilevel_equivalence}) is equivalent to the $q$-th sorted absolute residual value for $q = \lceil \tau n \rceil$. To this end, let us consider the ordered sequence of unique and different absolute residual values
$$0 \leq \bar s_{(1)} < \ldots < \bar s_{(i)} < \bar s_{(i+1)} < \ldots < \bar s_{(m)},$$

where $m \leq n$, such that for all $i \in \{1,...,m\}$ there exists a $\ell \in \{1,...,n\}$ for which $ s_{\ell} = \bar s_{(i)}$. Let us also define 
\begin{equation}
\label{eq:indicator_equivalence}
n_i := \sum_{\bar s_j \leq \bar s_{(i)}} 1 = \sum_{j=1}^{n} \mathbb{I}_{\{ \bar s_j \leq \bar s_{(i)} \}}.    
\end{equation}

Then, we have that for all $i \in \{1, ..., m-1\}$, 
\begin{align*}
  L_{\tau}(\bar s, \gamma) 
  & = (\tau-1) \sum_{\bar s_j \leq \bar s_{(i)}} (\bar s_j - \gamma) + \tau \sum_{ \bar s_j \geq \bar s_{(i+1)}} (\bar s_i - \gamma)  \\
  & = (\tau-1) \sum_{\bar s_j \leq \bar s_{(i)}} \bar s_j + \tau \sum_{ \bar s_j \geq \bar s_{(i+1)}} \bar s_j + \gamma (n_i - \tau n).
\end{align*}

For fixed $\bar{s}$,  $L_{\tau}(\bar s, \gamma)$ is a piecewise linear function of $\gamma$ with a breakpoint at $\bar s_{(i)}$. 
Therefore, we can observe that, for all $i \in \{1, ..., m\}$, $L_{\tau}(\bar s, \gamma)$ is decreasing if $n_{i} \leq \tau n$, and
increasing if $n_{i} \geq \tau n$.
Hence, $L_{\tau}(\bar s, \gamma)$ has only one minimum in $[0, + \infty)$
which corresponds to the absolute residue at position $q = \lceil \tau n \rceil$.
\end{proof}

Note that the right-hand side part of the equivalence in (\ref{eq:bilevel_equivalence}) has a bilevel problem structure. Based on this, a first bilevel approach is introduced where the upper-level problem sets the regression coefficients and the lower-level calls for the minimization of the sample quantile function (\ref{eq:tau_sample_quantile_function}), this is,
\begingroup
\allowdisplaybreaks
\begin{subequations}
\label{form:bilevel_1}
\begin{align}
  \hat{\beta}_4^{(\mathrm{LQS})} := & \arg \min_{\beta} \gamma        &                         \label{cons:BL_1_0} \\ \text{subject to: } 
  & r_{i}^{+} - r_{i}^{-} = y_{i}-\mathbf{x}_{i}^{\top} \beta,        & i \in \{1, \ldots, n\}, \label{cons:BL_1_1} \\
  & r_{i}^{+} + r_{i}^{-} = s_i,                                      & i \in \{1, \ldots, n\}, \label{cons:BL_1_2} \\
  & r_{i}^{+}, r_{i}^{-}, s_i \geq 0,                                 & i \in \{1, \ldots, n\}, \label{cons:BL_1_3} \\
  & (r_{i}^{+}, r_{i}^{-}): \SOS,                              & i \in \{1, \ldots, n\}, \label{cons:BL_1_4} \\
  & \gamma = \arg \min_{\gamma} \ L_{\tau}(s, \gamma).                 \label{cons:BL_1_5}
\end{align}
\end{subequations}
\endgroup

A single-level reformulation of (\ref{form:bilevel_1}) is provided in Proposition \ref{th:BL_1_bis}. 
\begin{proposition}
\label{th:BL_1_bis}

A solution to the $\mathrm{LQS}$ problem is given by 
\begingroup
\allowdisplaybreaks
\begin{subequations}
\label{form:BL_1_bis}
\begin{align}
\hat{\beta}_4^{(\mathrm{LQS})} := & \arg \min_{\beta} \quad \gamma                  & \label{cons:BL_1_0_bis} \\ \text{subject to: }
& r_{i}^{+}-r_{i}^{-}=y_{i}-\mathbf{x}_{i}^{\top} \beta,    & i \in \{1, \ldots, n\}, \label{cons:BL_1_1_bis} \\
& r_{i}^{+} + r_{i}^{-} = s_i,                              & i \in \{1, \ldots, n\}, \label{cons:BL_1_2_bis} \\
& \sum_{j=1}^{n} \mathbb{I}_{\{s_j \leq \gamma\}} \geq \lceil \tau n \rceil,       &  \label{cons:BL_1_3_bis} \\
& r_{i}^{+}, r_{i}^{-}, s_i \geq 0,                         & i \in \{1, \ldots, n\}, \label{cons:BL_1_4_bis} \\
& (r_{i}^{+}, r_{i}^{-}): \SOS,                             & i \in \{1, \ldots, n\}, \label{cons:BL_1_5_bis} \\
& \gamma \geq 0.                                            &                         \label{cons:BL_1_6_bis}
\end{align}
\end{subequations}
\endgroup

\end{proposition}

\begin{proof}
Constraint (\ref{cons:BL_1_3_bis}) is obtained as a consequence of the result in Lemma \ref{lem:bilevel_lemma}.
Given expression (\ref{eq:indicator_equivalence}) and that the inner minimization problem in (\ref{eq:bilevel_equivalence}) is equivalent to the $\lceil \tau n \rceil$-th sorted absolute residual value, i.e., $s_{(\lceil \tau n \rceil)}$, the equivalence
$$\min L_{\tau}(\bar s, \gamma) = \min \gamma, \text{ s.t.: } \sum_{j=1}^{n} \mathbb{I}_{\{s_j \leq \gamma\}} \geq \lceil \tau n \rceil$$

holds, which proves that adding condition (\ref{cons:BL_1_3_bis}) is equivalent to solving (\ref{cons:BL_1_5}).
\end{proof}

Note that the reformulation in Proposition \ref{th:BL_1_bis} is nonlinear because of constraints (\ref{cons:BL_1_3_bis}), which involve the sum of the indicator functions $\mathbb{I}_{\{s_i \leq \gamma\}}$, $i \in \{1, \ldots, n\}$. This set of constraints can be modeled by means of a set of binary variables $\kappa_i \in \{0,1\}$, $i \in \{1, \ldots, n\}$, which is 1 if $s_i \leq \gamma$, and 0 otherwise. Nevertheless, we keep the indicator description of these constraint since modern state-of-the-art general-purpose solvers allow to implement indicator functions using specific commands. 

A second bilevel approach is introduced where the lower-level calls for the minimization of the sample quantile function expressed as the optimization problem in (\ref{form:quantile_opt_problem}). Therefore, 
our second bilevel formulation is given by
\begingroup
\allowdisplaybreaks
\begin{subequations}
\label{form:bilevel_2}
\begin{align}
  \hat{\beta}^{(\mathrm{LQS})} := & \arg \min_{\beta} \gamma        &                             \label{cons:BL_2_0} \\ \text{subject to: } 
  & r_{i}^{+} - r_{i}^{-} = y_{i}-\mathbf{x}_{i}^{\top} \beta,        & i \in \{1, \ldots, n\},     \label{cons:BL_2_1} \\
  & r_{i}^{+} + r_{i}^{-} = s_i,                                      & i \in \{1, \ldots, n\},     \label{cons:BL_2_2} \\
  & s_i, r_{i}^{+}, r_{i}^{-} \geq 0,                                 & i \in \{1, \ldots, n\},     \label{cons:BL_2_3} \\
  & (r_{i}^{+}, r_{i}^{-}): \SOS,                                     & i \in \{1, \ldots, n\},     \label{cons:BL_2_4} \\
  & \gamma = \arg \min_{\gamma \in \mathbb{R}} \ (1-\tau) \sum_{i=1}^n u_i + \tau \sum_{i=1}^n v_i, \label{cons:BL_2_5} \\ 
  & \hspace{-1.00cm} \text{subject to: } \gamma - s_i \leq u_i,       & i \in \{1, \ldots, n\},     \label{cons:BL_2_6} \\
  & \hspace{ 0.75cm} s_i - \gamma \leq v_i,                           & i \in \{1, \ldots, n\},     \label{cons:BL_2_7} \\
  & \hspace{ 0.75cm} u_i, v_i \geq 0,                                 & i \in \{1, \ldots, n\},     \label{cons:BL_2_8} \\
  & \hspace{ 0.75cm} \gamma \geq 0.                                   &                             \label{cons:BL_2_9}
\end{align}
\end{subequations}
\endgroup

Since the lower-level problem (\ref{cons:BL_2_5})-(\ref{cons:BL_2_9}) is a linear optimization problem, one can apply the strong duality theorem to provide a single-level reformulation, which is presented in Proposition \ref{th:BL_2_bis}.

\begin{proposition}
\label{th:BL_2_bis}

A solution to the $\mathrm{LQS}$ problem is given by 

\begingroup
\allowdisplaybreaks
\begin{subequations}
\label{form:BL_2_bis}
\begin{align}
\hat{\beta}^{(\mathrm{LQS})} := & \arg \min_{\beta} \quad \gamma                & \label{cons:BL_2_0_bis} \\ \text{subject to: }
& r_{i}^{+}-r_{i}^{-}=y_{i}-\mathbf{x}_{i}^{\top} \beta,  & i \in \{1, \ldots, n\}, \label{cons:BL_2_1_bis} \\
& r_{i}^{+} + r_{i}^{-} = s_i,                            & i \in \{1, \ldots, n\}, \label{cons:BL_2_2_bis} \\
& (1-\tau) \sum_{i=1}^n u_i + \tau \sum_{i=1}^n v_i, = -\sum_{i=1}^{n} s_i \alpha_i + \sum_{i=1}^{n} s_i \beta_i, & \label{cons:BL_2_3_bis} \\
& \gamma - s_i \leq u_i,                                  & i \in \{1, \ldots, n\}, \label{cons:BL_2_4_bis} \\
& s_i - \gamma \leq v_i,                                  & i \in \{1, \ldots, n\}, \label{cons:BL_2_5_bis} \\
& \alpha_i \leq 1 - \tau,                                 & i \in \{1, \ldots, n\}, \label{cons:BL_2_6_bis} \\
& \beta_i \leq \tau,                                      & i \in \{1, \ldots, n\}, \label{cons:BL_2_7_bis} \\
& \sum_{i=1}^{n} \alpha_i - \sum_{i=1}^{n} \beta_i \leq  0, & i \in \{1, \ldots, n\}, \label{cons:BL_2_8_bis} \\
& s_i, r_{i}^{+}, r_{i}^{-} \geq 0,                       & i \in \{1, \ldots, n\}, \label{cons:BL_2_9_bis} \\
& (r_{i}^{+}, r_{i}^{-}): \SOS,                           & i \in \{1, \ldots, n\}, \label{cons:BL_2_10_bis} \\
& u_i, v_i \geq 0,                                        & i \in \{1, \ldots, n\}, \label{cons:BL_2_11_bis} \\
& \alpha_i, \beta_i \geq 0,                               & i \in \{1, \ldots, n\}, \label{cons:BL_2_12_bis} \\
& \gamma \geq 0.                                          &                         \label{cons:BL_2_13_bis}
\end{align}
\end{subequations}
\endgroup

\end{proposition}

\begin{proof}
Given the variables $\bar{s}_i\geq 0, i\in \{1,...,n\}$, fixed in the upper-level problem, the dual problem of the lower-level problem (\ref{cons:BL_2_5})-(\ref{cons:BL_2_9}) is given by

\begingroup
\allowdisplaybreaks
\begin{subequations}
\label{form:BL_2_lldual}
\begin{align}
& \min_{\beta} \quad -\sum_{i=1}^{n} \bar{s}_i \alpha_i + \sum_{i=1}^{n} \bar{s}_i \beta_i & \label{cons:BL_2_lldual_0} \\ \text{subject to: }
& \alpha_i \leq 1 - \tau,                                          & i \in \{1, \ldots, n\}, \label{cons:BL_2_lldual_1} \\
& \beta_i \leq \tau,                                               & i \in \{1, \ldots, n\}, \label{cons:BL_2_lldual_2} \\
& \sum_{i=1}^{n} \alpha_i - \sum_{i=1}^{n} \beta_i \leq 0,         & i \in \{1, \ldots, n\}, \label{cons:BL_2_lldual_3} \\
& \alpha_i, \beta_i \geq 0,                                        & i \in \{1, \ldots, n\}. \label{cons:BL_2_lldual_4}
\end{align}
\end{subequations}
\endgroup

Therefore, problem (\ref{form:BL_2_bis}) is obtained as single-level reformulation of (\ref{form:bilevel_2}) by applying strong-duality theorem, which is given by the constraints of the upper-level problem and primal and dual feasibility and the equality of the primal and dual objective functions of the lower-level problems.
\end{proof}

Note that formulation (\ref{form:BL_2_bis}) is also nonlinear because of the products of the non-negative variables representing the absolute value of the residuals and the non-negative dual variables, which cannot be linearized via McCormick envelopes.
For that reason,
although the model has methodological interest, it is clearly outperformed by the previous one and is of limited applicability from a computational perspective.

\section{Computational experience}
\label{sec:Computational_experience}

This section summarizes the results obtained from our computational experiments performed in order to empirically compare the proposed LQS formulations. 
The performance of the formulation to obtain $\widehat{\beta}_{BM}$ estimators is compared with the formulations for $\widehat{\beta}_{1}$, $\widehat{\beta}_{2}$, $\widehat{\beta}_{3}$ and $\widehat{\beta}_{4}$ estimators. 
We will overuse notation by calling the formulation that computes the $\beta_{(.)}$ estimator by $\beta_{(.)}$.
Our formulations were programmed in Python language (version 3.10) using Gurobi Optimizer 9.5 as a MIP solver. The experiments were run in a MacPro server with a 2,7 GHz Intel Xeon W processor of 24 cores and 192 GB RAM, using 8 of those threads per each run. 
The time limit for computations was set to 7200 seconds per instance.
Following the computational experiments in \cite{BertsimasMazumder2014}, we test our formulations in two set of instances: real-world datasets and synthetic datasets. 

Regarding real-world datasets, we considered two set of instances available from the \texttt{R}-package \texttt{robustbase} (\citealp{ Todorov2009, Rousseeuw2024}).
First, the \texttt{Alcohol} dataset, which is aimed to study the solubility of alcohols in water to understand alcohol transport in living organisms. This dataset contains physicochemical characteristics of $n = 44$ aliphatic alcohols and measurements on seven numeric variables: SAG solvent accessible surface-bounded molecular volume ($\texttt{x}_1$), logarithm of the octanol-water partitions coefficient ($\texttt{x}_2$), polarizability ($\texttt{x}_3$), molar refractivity ($\texttt{x}_4$), mass ($\texttt{x}_5$), volume ($\texttt{x}_6$) and the response (\texttt{y}) is taken to be the logarithm of the solubility. 
We consider two cases from the dataset, a first one with $n = 44$ and $p = 5$ where the five covariates were $\texttt{x}_1, \texttt{x}_2, \texttt{x}_4, \texttt{x}_5, \texttt{x}_6$, and a second that has all the six covariates and an intercept term, which leads to $p = 7$. 
In this example, the LQS estimator was computed for $q = 31$.
The \texttt{HBK} dataset created by \cite{Hawkins1984} is also considered in this experimentation set-up. 
The dataset consists of 75 observations with one response and three explanatory variables, that is, $n = 75$ and $p = 3$. In this example, the LQS estimator was computed $q \in \{60, 45\}$. 

\rv{
Table \ref{tab:realworld} shows results from real-world datasets. The performance of each formulation is evaluated for every instance using four metrics: \texttt{Time} (time to compute the best upper and lower bounds), \texttt{UB} (best upper bound), \texttt{LB} (best lower bound), and \texttt{TimeUB} (time to find the best upper bound).
Generally, instances become harder to solve as the number of observations and covariates increases or as the quantile decreases. Across all instances, the formulations showed similar performance overall.
The $\widehat{\beta}_{4}$ formulation performed best in terms of \texttt{Time} and \texttt{TimeUB}, followed by the $\widehat{\beta}_{2}$ formulation. In contrast, $\widehat{\beta}_{3}$ performed worse, especially with larger datasets. It could not certify optimality (i.e., close the gap between the best upper and lower bounds) in any instance. However, it performed well on \texttt{TimeUB} for the \texttt{Alcohol} dataset.
}

Regarding synthetic datasets, following \cite{Rousseeuw2006}, the model matrix $\mathbf{X}_{n \times p}$ was generated with i.i.d. Gaussian entries $N(0,100)$ and $\beta \in \mathbb{R}^p$ was taken to be a vector of all ones. 
Subsequently, the response was generated as $\mathbf{y}=\mathbf{X} \beta+\varepsilon$, where $\varepsilon_i \sim N(0,10)$, $i=\{1, \ldots, n\}$. Once $(\mathbf{y}, \mathbf{X})$ have been generated, a certain proportion $\pi$ of the samples is corrupted in two different ways:

\begin{enumerate}[(A)] 
\item $\lfloor\pi n\rfloor$ of the samples are chosen at random and the first coordinate of the data matrix $\mathbf{X}$, that is, $x_{1j}$ is replaced by $x_{1j} \leftarrow x_{1j}+1000$, hence outliers are only added in the covariate space.
\item $\lfloor\pi n\rfloor$ of the samples are chosen at random out of which the covariates of half of the points are changed as in item (A). For the remaining half of the points the responses are corrupted as $y_j \leftarrow y_j+1000$, so outliers are added in both the covariate and response spaces.
\end{enumerate}

\hspace{1cm}

\begin{footnotesize}
\begin{longtable}[c]{c|rrrr|rrrr}
\toprule
\texttt{Estimator} & \texttt{Time} & \texttt{UB} & \texttt{LB} & \texttt{TimeUB} & \texttt{Time} & \texttt{UB} & \texttt{LB} & \texttt{TimeUB} \\ \midrule
\texttt{Alcohol} & \multicolumn{4}{c|}{(n = 44, p = 5, q = 31)} & \multicolumn{4}{c}{(n = 44, p = 7, q = 31)}       \\ \midrule 
$\widehat{\beta}_{BM}$ &     27.79 & 0.196 & 0.196 &  9.86  &   729.54 & 0.156 & 0.156 &  51.96  \\
$\widehat{\beta}_1$    &     15.41 & 0.196 & 0.196 &  2.96  &    60.29 & 0.156 & 0.156 &   6.57  \\
$\widehat{\beta}_2$    &      3.20 & 0.196 & 0.196 &  0.25  &    22.87 & 0.156 & 0.156 &   4.39  \\
$\widehat{\beta}_3$    &   7200.00 & 0.196 & 0.068 &  2.51  &  7200.00 & 0.156 & 0.046 &   0.08  \\
$\widehat{\beta}_4$    &      1.78 & 0.196 & 0.196 &  0.21  &    13.39 & 0.156 & 0.156 &   0.55  \\ \midrule
\texttt{HBK} & \multicolumn{4}{c|}{(n = 75, p = 3, q = 45)} & \multicolumn{4}{c}{(n = 75, p = 3, q = 60)}      \\ \midrule 
$\widehat{\beta}_{BM}$ &    55.93 & 0.585 & 0.585 &   48.59  &     1.01 & 0.819 & 0.819 &   0.91  \\
$\widehat{\beta}_1$    &  1874.29 & 0.585 & 0.585 &    1.43  &    21.92 & 0.819 & 0.819 &   2.13  \\
$\widehat{\beta}_2$    &    30.73 & 0.585 & 0.585 &   26.80  &     0.44 & 0.819 & 0.819 &   0.43  \\
$\widehat{\beta}_3$    &  7200.00 & 0.612 & 0.054 & 7200.00  &  7200.00 & 0.819 & 0.107 &   5.94  \\
$\widehat{\beta}_4$    &     4.79 & 0.585 & 0.585 &    0.37  &     0.33 & 0.819 & 0.819 &   0.13  \\ \bottomrule
\caption{Real-world datasets results}
\label{tab:realworld}\\
\end{longtable}
\end{footnotesize}

\rv{This provides two datasets in which outliers are artificially generated with outliers only in covariate space (type A) and with outliers in both covariate and response space (type B), which are more difficult to solve.}
Following \cite{BertsimasMazumder2014} experimentation set-up, to compare the different models in terms of the quality of solutions obtained we proceeded as follows. 
For every instance, all formulations are run within the time limit and the best solution $\texttt{f}_{*}$ among them is obtained. 
If $\texttt{f}_{m}$ denotes the value of the LQS objective function for formulation "m", then the relative accuracy of the solution obtained by "m" is defined as
$$\texttt{Relative Accuracy} =\left(\texttt{f}_{m} - \texttt{f}_* \right) / \texttt{f}_* \times 100.$$

\rv{
Tables \eqref{tab:instances50}, \eqref{tab:instances75} and \eqref{tab:instances101}, show the performance of the formulations in instances of 50, 75 and 101 data points, respectively.
For every ($n,p,q,\pi,\text{perturbation type}$) combination, the column \texttt{Accuracy} accounts for the \texttt{Relative Accuracy} and the numbers within parenthesis denote the standard errors. Results are averaged over 20 different random instances of the problem. 
The tables have been arranged so that the left column shows the results for instances with type A perturbation and in the right column with type B.
Similarly, Figures \ref{fig:syn_1} and \ref{fig:syn_2} depict average times and relative accuracy percentages for the different data sizes. 
Figures concerning averages times depict both, the time elapsed by the formulation for the computation of the best upper bound (green) and the time elapsed to finish solving the instance (yellow).
Therefore, the yellow stripes identify the time that the formulations spend on average to close the gap by making the effort to raise the lower bound. 
One can see that,
for instances with 50 and 75 data points, the best-performing formulations are $\widehat{\beta}_{2}$ and $\widehat{\beta}_{4}$. 
Formulation $\widehat{\beta}_{2}$ performs better for type A perturbations, while $\widehat{\beta}_{4}$ excels for type B.
For instances with 101 data points, $\widehat{\beta}_{4}$ is the top performer, particularly in large instances where it can certify optimality within the time limit, unlike the other formulations. 
It is followed by $\widehat{\beta}_{2}$.
The only scenario where $\widehat{\beta}_{BM}$ outperforms $\widehat{\beta}_{4}$ and $\widehat{\beta}_{2}$ in terms of \texttt{Time} is with the parameter combination $(75, 5, 40, 0.4, A)$. However, in this case, $\widehat{\beta}_{4}$ and $\widehat{\beta}_{2}$ compute the best upper bound much faster, outperforming $\widehat{\beta}_{BM}$ in terms of \texttt{TimeUB}.
Formulation $\widehat{\beta}_{1}$ occasionally outperforms $\widehat{\beta}_{2}$ and $\widehat{\beta}_{4}$ for \texttt{TimeUB}, such as in $(75, 10, 25, 0.5, B)$ instances. Meanwhile, $\widehat{\beta}_{3}$ consistently shows the worst performance across all metrics.
}

\section{Scalability}
\label{sec:scale}

In this section, we address the scalability of the problem to larger size instances (up to 10000 points) by introducing a novel aggregation approach.
Let us denote by $X$ the original input data and by $\hat X$ the multiset of the aggregate data. For each $x\in X$ we associate $\hat x \in \hat X$, possibly several $x\in X$ correspond to the same $\hat x \in \hat X$, so that the cardinality of the different elements in $\hat X$ is much smaller than $| X|$. 
Our goal is to provide some reduction in size result for which the error bound of the quantile regression is bounded and asymptotically converges to zero whenever $\hat X \rightarrow X$. 
For the sake of readability, let us denote $f_q(X, \beta):=|r_{(q)}(X,\beta)|$, the $q$-th quantile of the absolute residuals computed with respect to the data set $X$ and the hyperplane defined by the vector $\beta$.
The main result in this section is the following theorem.

\begin{theorem}
\label{th:scale}
    Let $D=\max_{x\in X} |x-\hat x|$ and $\beta^*=\arg \min _{\beta} f_q(X,\beta)$, $\hat \beta=\arg \min _{\beta} f_q(\hat X,\beta)$, the optimal $q$-quantile hyperplanes with respect to the data sets $X$ and $\hat X$, respectively. Then, 
    \begin{equation}
        \label{eq:scala} |f_q(X,\beta^*)-f_q(X,\hat \beta)|\le 2 D
    \end{equation}
\end{theorem}

\begin{footnotesize}
\begin{longtable}[c]{c|rrr|rrr}
\toprule
\texttt{Estimator} & \texttt{Time} & \texttt{Accuracy} & \texttt{TimeUB} & \texttt{Time} & \texttt{Accuracy} & \texttt{TimeUB} \\ \cmidrule{2-7} 
& \multicolumn{3}{c|}{50 - 5 - 15 - 0.4 - A} & \multicolumn{3}{c}{50 - 5 - 15 - 0.4 - B}       \\ \midrule 
$\widehat{\beta}_{BM}$ &  112.74 & 0.00 (0.00) &   65.89  &  121.52 & 0.00 (0.00) &   73.8     \\
$\widehat{\beta}_1$    &  285.01 & 0.00 (0.00) &   81.16  &  274.37 & 0.00 (0.00) &  109.61    \\
$\widehat{\beta}_2$    &   79.62 & 0.00 (0.00) &   23.04  &  146.63 & 0.00 (0.00) &   85.05    \\
$\widehat{\beta}_3$    &  498.58 & 0.00 (0.00) &  139.35  & 1184.47 & 0.00 (0.00) &  337.51    \\
$\widehat{\beta}_4$    &   98.63 & 0.00 (0.00) &   30.24  &  117.10  & 0.00 (0.00) &   42.69   \\ \midrule
& \multicolumn{3}{c|}{50 - 5 - 30 - 0.4 - A} & \multicolumn{3}{c}{50 - 5 - 30 - 0.4 - B}       \\ \midrule
$\widehat{\beta}_{BM}$ &    8.98 & 0.00 (0.00) &   8.11 &    210.51 & 0.00 (0.00) &  52.47     \\
$\widehat{\beta}_1$    &   23.89 & 0.00 (0.00) &  12.33 &     66.53 & 0.00 (0.00) &  15.44     \\
$\widehat{\beta}_2$    &    5.19 & 0.00 (0.00) &   0.92 &     30.96 & 0.00 (0.00) &   1.41     \\
$\widehat{\beta}_3$    & 6810.73 & 0.00 (0.00) & 107.92 &   7200.00 & 0.04 (0.04) & 710.26     \\
$\widehat{\beta}_4$    &    8.13 & 0.00 (0.00) &   4.72 &      9.17 & 0.00 (0.00) &   1.11     \\ \midrule
& \multicolumn{3}{c|}{50 - 10 - 15 - 0.5 - A} & \multicolumn{3}{c}{50 - 10 - 15 - 0.5 - B}     \\ \midrule
$\widehat{\beta}_{BM}$ & 7200.00 & 55.94 (12.98) & 6098.87 & 7200.00 & 68.47 (16.08) & 6879.45 \\
$\widehat{\beta}_1$    & 7200.00 & 59.68 (13.04) & 5244.64 & 7200.00 & 96.32 (20.10) & 6809.62 \\
$\widehat{\beta}_2$    & 7200.00 & 32.57 ( 9.93) & 5265.93 & 7200.00 & 44.14 (12.12) & 5843.09 \\
$\widehat{\beta}_3$    & 7200.00 & 63.08 (11.97) & 7008.88 & 7200.00 & 56.30 (13.78) & 6996.63 \\
$\widehat{\beta}_4$    & 7200.00 & 39.74 (10.13) & 5966.24 & 7200.00 & 34.91 ( 8.86) & 6232.40 \\ \midrule
& \multicolumn{3}{c|}{50 - 10 - 30 - 0.5 - A} & \multicolumn{3}{c}{50 - 10 - 30 - 0.5 - B}     \\ \midrule
$\widehat{\beta}_{BM}$ & 6309.04 &  4.07 (2.14) & 4223.32 & 7139.45 & 0.49 (0.29) &  2719.85   \\
$\widehat{\beta}_1$    & 7148.67 &  2.14 (1.21) & 3113.89 & 6813.67 & 0.00 (0.00) &  1797.70   \\
$\widehat{\beta}_2$    & 5315.74 &  0.64 (0.64) & 1722.19 & 6504.35 & 0.35 (0.24) &  1919.68   \\
$\widehat{\beta}_3$    & 7200.00 &  7.22 (2.19) & 5412.28 & 7200.00 & 5.47 (1.50) &  5135.51   \\
$\widehat{\beta}_4$    & 5052.41 &  0.17 (0.14) & 1586.56 & 3055.81 & 0.00 (0.00) &   606.24   \\ \bottomrule
\caption{Synthetic data results for n = 50}
\label{tab:instances50}
\end{longtable}
\end{footnotesize}

\begin{figure}[h]
    \centering
    \includegraphics[width=15cm]{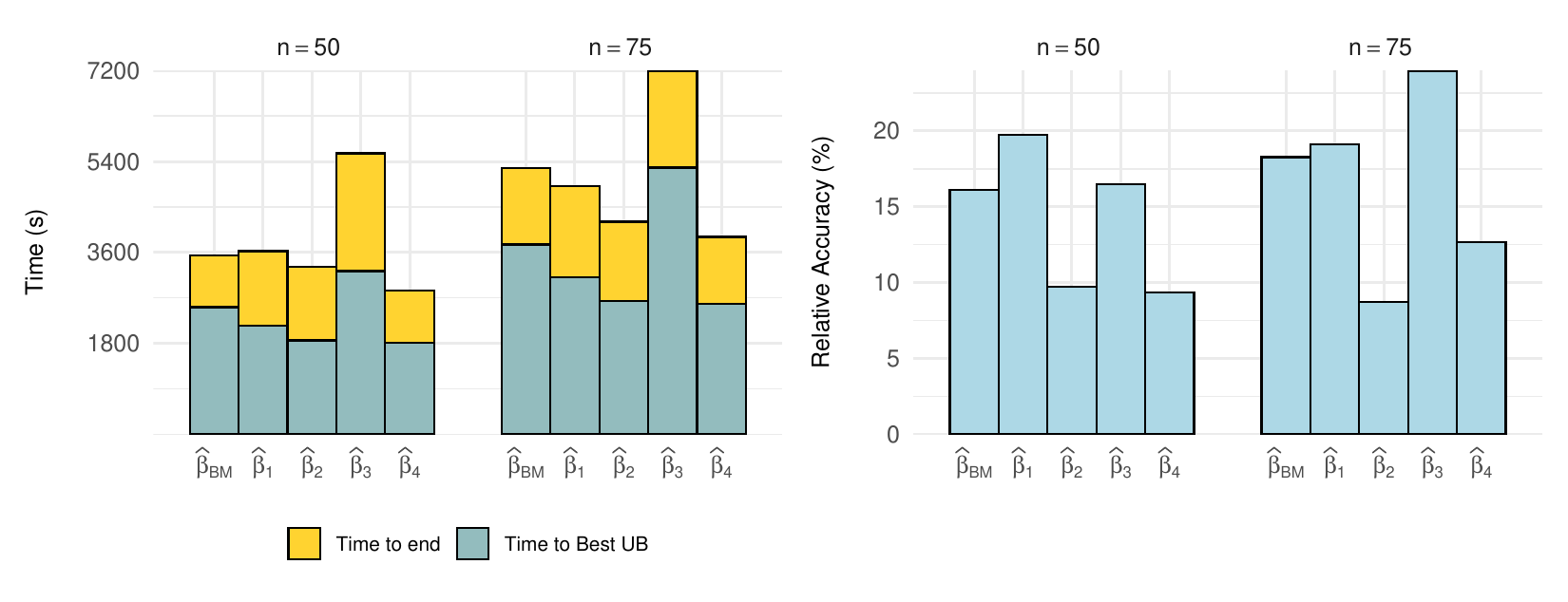}
    \caption{Average times and relative accuracy percentages results for n=25,75}
    \label{fig:syn_1}
\end{figure}

\begin{proof}
Recall that $ r_i(\beta)= y_i-\mathbf{x}_i^{\top} \beta$ and let $e_X(\beta)=(|r_i(\beta)|)_{i=1}^{n}$. 
First of all, we observe that by the triangle inequality it always holds
$$ e_X(\beta)\le e_{\hat X} (\beta) +D \underline{1}^t, $$
being $D=\max_i \{|x_i-\hat x_i|\}$. Let $N=\{1,\ldots,n\}$.
Then, since the $q$-th residual is an isotonic operation 
\begin{align*}
    f_r(X,\beta)=r_{(q)}(e_X(\beta))\le & r_{(q)}(e_{\hat X} (\beta) +D \underline{1}^t)\\
   \le &\min_{\mathcal{I}\subset N: |\mathcal{I}|=q} \max _{j\in \mathcal{I}} \{|e_{\hat x_j}|+D\}\\
  \le  & \max _{j\in \mathcal{I}} \{|e_{\hat x_j}|+D\},\; \forall {\mathcal{I}}\subset N: |{\mathcal{I}}|=q
\end{align*}
Thus, $f_q(\hat X,\beta)\le f_q(X, \beta) + D$ for all $\beta$.
Analogously, $f_q(X,\beta)\le f_q(\hat X, \beta) + D$ for all $\beta$.

Next, the main inequality follows by the following 
    \begin{align*}
    f_q(X,\hat \beta) \le f_q(\hat X,\hat \beta)+D\le f_q(\hat X, \beta^*)+D \le f_q(X, \beta^*)+2D \\
    f_q(X,\hat \beta) \ge f_q(X,\beta^*)\ge f_q(\hat X,\beta^*)-D \ge f_q( X,\beta^*)-2D 
    \end{align*}
    
Joining both inequalities, one gets:
$$ |f_q(X,\beta^*)-f_q(X,\hat \beta)|\le 2 D.$$
\end{proof}

\begin{footnotesize}
\begin{longtable}[h]{c|rrr|rrr}
\toprule
{\texttt{Estimator}} & \texttt{Time} & \texttt{Accuracy} & \texttt{TimeUB} & \texttt{Time} & \texttt{Accuracy} & \texttt{TimeUB} \\ \cmidrule{2-7} 
& \multicolumn{3}{c|}{ 75 - 5 - 25 - 0.4 - A } & \multicolumn{3}{c}{ 75 - 5 - 25 - 0.4 - B }        \\ \midrule 
$\widehat{\beta}_{BM}$ &  1633.62 &     0.00 (0.00) &  724.56 & 4430.51 &  0.00  (0.00) &  1247.38  \\
$\widehat{\beta}_1$    &  2578.67 &     0.00 (0.00) &  252.49 & 3913.49 &  0.00  (0.00) &  1533.05  \\
$\widehat{\beta}_2$    &  1314.71 &     0.00 (0.00) &  451.75 & 2552.69 &  0.00  (0.00) &  1628.06  \\
$\widehat{\beta}_3$    &  7200.00 &     1.55 (0.71) & 3587.03 & 7200.00 &  3.73  (0.90) &  4946.81  \\
$\widehat{\beta}_4$    &   959.56 &     0.00 (0.00) &  407.40 &  958.24 &  0.00  (0.00) &   350.15  \\ \midrule
& \multicolumn{3}{c|}{ 75 - 5 - 40 - 0.4 - A } & \multicolumn{3}{c}{ 75 - 5 - 40 - 0.4 - B }        \\ \midrule
$\widehat{\beta}_{BM}$ &   104.06 &     0.00 (0.00) &   39.29 & 7200.00 &  2.05  (1.14) & 3080.55   \\
$\widehat{\beta}_1$    &   391.33 &     0.00 (0.00) &   25.20 & 3627.74 &  0.00  (0.00) &  538.61   \\
$\widehat{\beta}_2$    &   214.89 &     0.00 (0.00) &    1.73 &  780.33 &  0.00  (0.00) &  254.48   \\
$\widehat{\beta}_3$    &  7200.00 &     0.00 (0.00) & 1396.72 & 7200.00 &  4.40  (1.18) & 5630.87   \\
$\widehat{\beta}_4$    &   166.03 &     0.00 (0.00) &    6.43 &  373.74 &  0.00  (0.00) &   84.35   \\ \midrule
& \multicolumn{3}{c|}{ 75 - 10 - 25 - 0.5 - A } & \multicolumn{3}{c}{ 75 - 10 - 25 - 0.5 - B }      \\ \midrule
$\widehat{\beta}_{BM}$ &  7200.00 &  60.78  (16.85) & 5780.45 & 7200.00 & 52.43  (9.33) & 6485.96   \\
$\widehat{\beta}_1$    &  7200.00 &  95.97  (24.20) & 5634.86 & 7200.00 & 31.94  (7.90) & 5616.51   \\
$\widehat{\beta}_2$    &  7200.00 &  41.26  (14.00) & 4756.77 & 7200.00 & 20.96  (4.77) & 6232.54   \\
$\widehat{\beta}_3$    &  7200.00 & 117.97  (23.62) & 6923.41 & 7200.00 & 31.73  (5.58) & 6888.15   \\
$\widehat{\beta}_4$    &  7200.00 &  54.13  (15.09) & 4446.04 & 7200.00 & 31.59  (7.01) & 6076.58   \\ \midrule
& \multicolumn{3}{c|}{ 75 - 10 - 40 - 0.5 - A } & \multicolumn{3}{c}{ 75 - 10 - 40 - 0.5 - B }      \\ \midrule
$\widehat{\beta}_{BM}$ & 7200.00 &  10.72   (2.06) & 6084.50 & 7200.00 & 20.17 (3.82) & 6606.25     \\
$\widehat{\beta}_1$    & 7200.00 &   5.43   (1.41) & 4606.60 & 7200.00 & 19.45 (3.40) & 6663.73     \\
$\widehat{\beta}_2$    & 7200.00 &   2.57   (1.18) & 3347.07 & 7200.00 &  4.92 (1.54) & 4485.96     \\
$\widehat{\beta}_3$    & 7200.00 &  16.33   (3.17) & 6390.42 & 7200.00 & 16.13 (3.25) & 6490.46     \\
$\widehat{\beta}_4$    & 7200.00 &   2.63   (0.97) & 3645.66 & 7200.00 & 12.85 (2.66) & 5632.45     \\ \bottomrule
\caption{Synthetic data results for n = 75}
\label{tab:instances75}
\end{longtable}
\end{footnotesize}

\begin{footnotesize}
\begin{longtable}[c]{c|rrr|rrr}
\toprule
{\texttt{Estimator}} & \texttt{Time} & \texttt{Accuracy} & \texttt{TimeUB} & \texttt{Time} & \texttt{Accuracy} & \texttt{TimeUB} \\ \cmidrule{2-7} 
& \multicolumn{3}{c|}{ 101 - 5 - 51 - 0.4 - A } & \multicolumn{3}{c}{ 101 - 5 - 51 - 0.4 - B }    \\ \midrule 
$\widehat{\beta}_{BM}$ &   2182.81  &  0.00 (0.00)  &  866.75 & 7200.00 & 1.98 (0.88)  & 3671.47  \\
$\widehat{\beta}_1$    &   3443.78  &  0.00 (0.00)  &  546.49 & 7098.26 & 1.62 (0.66)  & 4538.83  \\
$\widehat{\beta}_2$    &   2116.53  &  0.00 (0.00)  &  366.63 & 7200.00 & 0.16 (0.13)  & 1518.30  \\
$\widehat{\beta}_3$    &   7200.00  &  0.74 (0.51)  & 3334.87 & 7200.00 & 7.84 (1.58)  & 6169.86  \\
$\widehat{\beta}_4$    &   1289.68  &  0.00 (0.00)  &   21.03 & 5454.79 & 0.00 (0.00)  & 1303.71  \\ \midrule
& \multicolumn{3}{c|}{ 101 - 5 - 75 - 0.4 - A } & \multicolumn{3}{c}{ 101 - 5 - 75 - 0.4 - B }    \\ \midrule
$\widehat{\beta}_{BM}$ &   227.67   &  0.00 (0.00)  &  118.18 & 6588.01 & 0.05 (0.05)  & 3428.24  \\
$\widehat{\beta}_1$    &  1353.86   &  0.00 (0.00)  &  185.54 & 2295.74 & 0.00 (0.00)  &  665.38  \\
$\widehat{\beta}_2$    &   477.77   &  0.00 (0.00)  &   29.47 & 4842.17 & 0.00 (0.00)  &  407.00  \\
$\widehat{\beta}_3$    &  7200.00   &  0.69 (0.24)  & 3120.06 & 7200.00 & 1.44 (0.40)  & 4762.61  \\
$\widehat{\beta}_4$    &   273.06   &  0.00 (0.00)  &   17.28 &  467.72 & 0.00 (0.00)  &   83.43  \\ \bottomrule
\caption{Synthetic data results for n = 101}
\label{tab:instances101}\\
\end{longtable}
\end{footnotesize}

\begin{figure}[h]
    \centering
    \includegraphics[width=15cm]{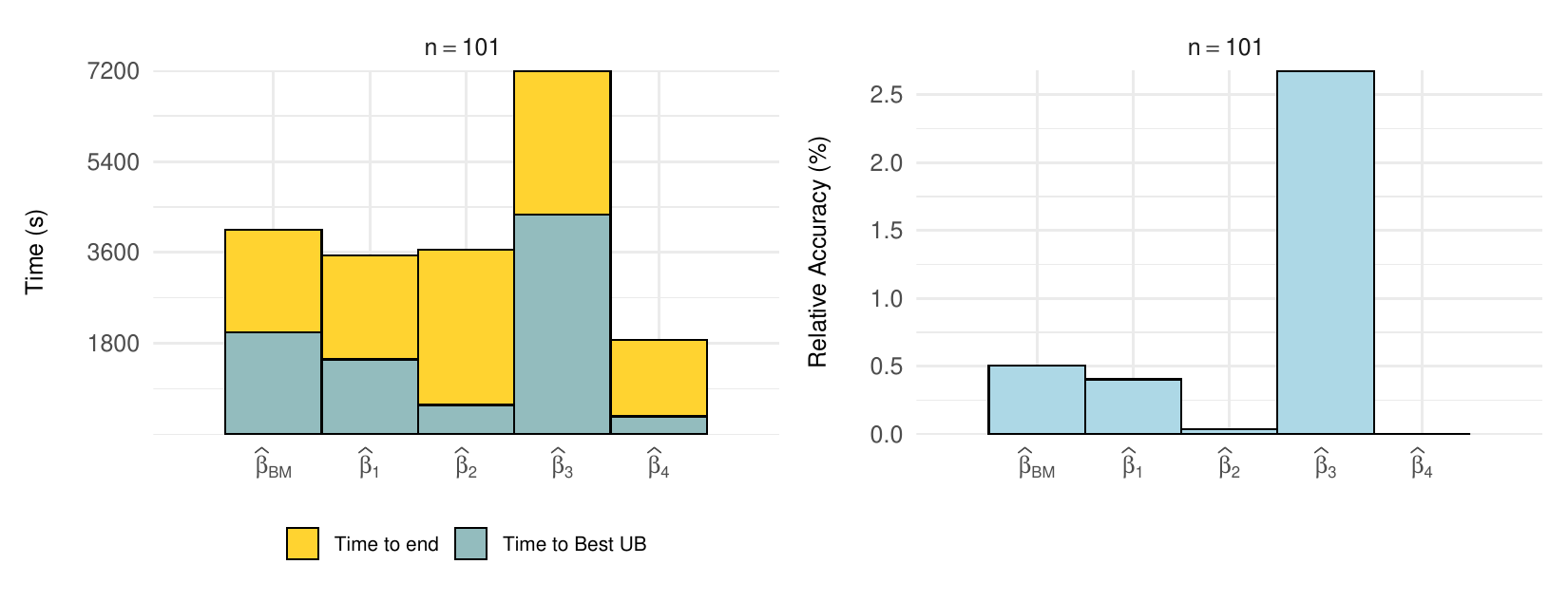}
    \caption{Average times and relative accuracy percentages results for n=101}
    \label{fig:syn_2}
\end{figure}

We note in passing that whenever $\hat X \rightarrow X$ the parameter $D\rightarrow 0$ and the approximation converges to the actual optimal solution.
Therefore, our approach works as follows. 
The main idea is to aggregate our dataset into a tractable number of points or clusters for the exact formulations previously presented, that is, to reduce the datasets to a size where the formulations can provide a solution in an admissible time limit.
The clusters formed must be as homogeneous as possible, i.e., the number of original data points in each cluster must be similar. 

Given the above conditions and the result in Theorem \ref{th:scale}, an approximate LQS estimator can be obtained by solving the aggregated dataset by means of an adjusted version of the above formulations where the weight of each cluster is included.
For example, let $k$ be the number of clusters considered, $N_i$ the number of data points in cluster $i$, and $C_i$ represents the indexes belonging to cluster $i$,
then formulation (\ref{form:SL_1_bis_lin}) can be rewritten modifying constraints (\ref{cons:SL_2_4_bis_lin}) and (\ref{cons:SL_2_5_bis_lin}) to take into account the weight of each cluster as follows,
\begingroup
\allowdisplaybreaks
\begin{subequations}
\label{form:SL_1_bis_lin}
\begin{align}
    & \sum_{i=1}^{k} N_i z_i \leq n-q,    &                                   \label{cons:SL_2_4_bis_mod} \\
    & s_j \leq \omega + M z_i,            & i \in \{1, \ldots, k\}, j \in C_i \label{cons:SL_2_5_bis_mod} 
\end{align}
\end{subequations}
\endgroup

Figure \ref{fig:SCALE} depicts the performance of our approach, namely \texttt{SCALE}, comparing the solutions with those obtained by the LQS algorithm from the \texttt{R}-package \texttt{MASS} (\citealp{MASS}).
Following \cite{BertsimasMazumder2014} experimentation set-up for large instances,
we test 20 instances of the type $(2001,10,1201,0.4,B)$, $(5001,10,3001,0.4,B)$ and $(10001,20,6001,0.4,B)$. 
We aggregate each of the original datasets to 100 representative data points or clusters, letting our approach run for a time limit of 7200 seconds.
It can be seen how our approach, \texttt{SCALE}, provides better solutions than those computed using the \texttt{MASS} package, especially, the difference is more remarkable the larger the instance size solved. 
\begin{figure}[h]
    \centering
    \includegraphics[width=15cm]{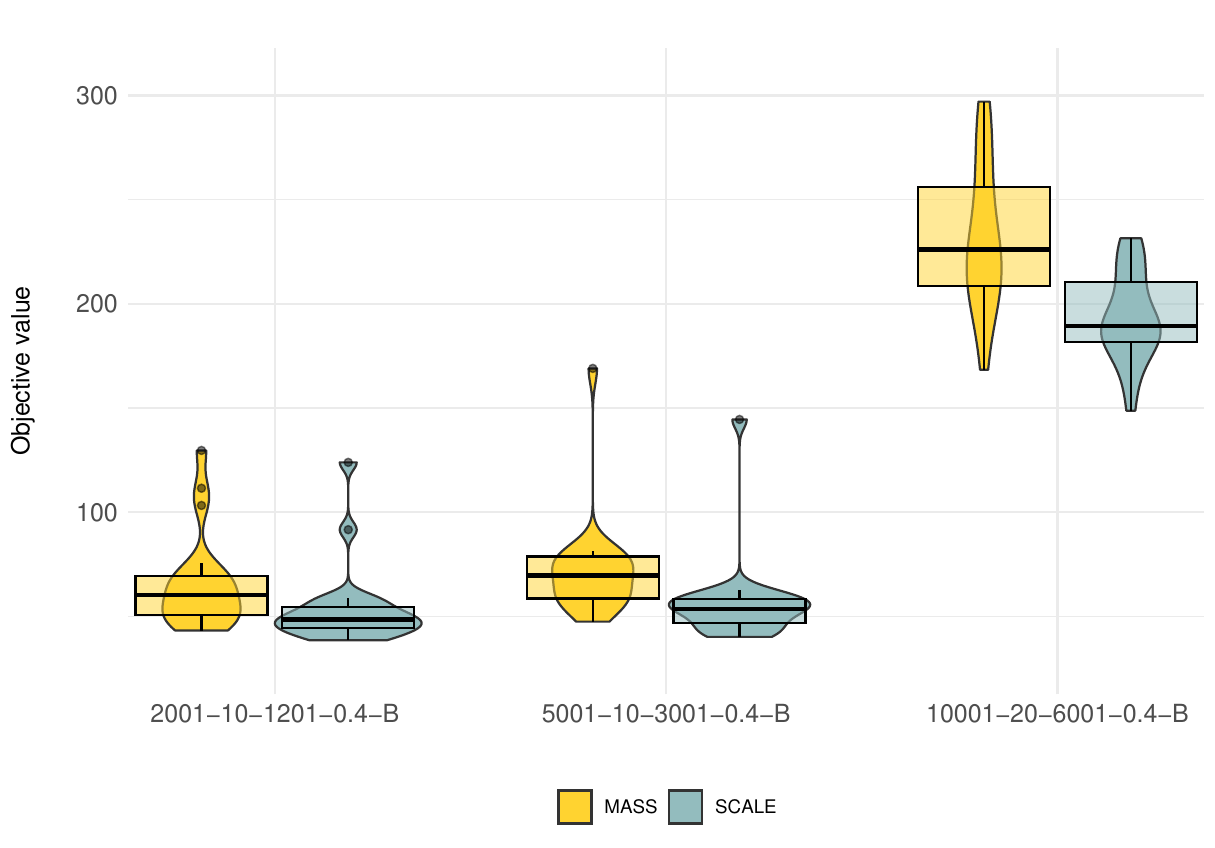}
    \caption{Comparison between LQS solutions obtained by the \texttt{R}-package \texttt{MASS} and our approach, namely \texttt{SCALE}.}
    \label{fig:SCALE}
\end{figure}

\section{Conclusions}

In this paper, we propose algorithms for LQS problems
based on optimization with ordering and single-and-bilevel  mathematical programming. We compare several formulations based on different rationale showing their usefulness.
\begin{enumerate}
    \item Our exact formulations solve to optimality the problem of finding the quantile regression coefficients and the minimum $r_{(q)}$ absolute residual of any data set with size up to 100 in less than 2 hours. For larger sizes, our methods provide rather good upper bounds in reasonable time. Actually, we present an aggregation approximate algorithm that outperforms state-of-the art methods for quantile regression (\texttt{R}-package \texttt{MASS} \citealp{MASS}) for data in the range of thousands ($n=10000$).
    \item Our formulations outperform the only available exact method for quantile regression \cite{BertsimasMazumder2014}, for those problem sizes where optimality can be certified. 
    \item Our estimators inherit the same properties that those obtained by \cite{BertsimasMazumder2014} in terms of validity for data point  not necessarily in general position, high breakdown point and capability to incorporate conic and polyhedral constraints on the regression coefficients. 
    \item The approaches presented in this paper are rahter flexible to accommodate to many other forms of robust regression models. Indeed, they easily extend to more sophisticated forms of regression which are based on ordered weighted averages of absolute residuals. These regression coefficients can be optimally obtained solving their corresponding exact formulations. The interested reader is referred to \cite{Nickel2005},\cite{Labbe2017} or \cite{Marin2020}, among other sources.
\end{enumerate}

\section*{\small Acknowledgements}

\begin{small}
The authors acknowledge financial support by the reasearch project PID2020-114594GB-C21 funded by the Spanish Ministerio de Ciencia e Innovación and Agencia Estatal de Investigación (MCIN/AEI/10.13039/501100011033).     
\end{small}

\bibliographystyle{elsarticle-harv} 
\bibliography{refs}

\end{document}